\RequirePackage{amsmath}
\documentclass{elsarticle}
\usepackage[utf8]{inputenc}
\usepackage[T1]{fontenc}
\usepackage{latexsym}
\usepackage{fix-cm}
\usepackage{amssymb}
\usepackage{array}

\usepackage{url}
\usepackage{hyperref}
\usepackage{lineno}
\usepackage{graphicx}
\usepackage{tabularx}
\usepackage{array}
\usepackage{mdwmath}
\usepackage{mdwtab}
\usepackage{float}

\usepackage[cmyk]{xcolor}
\usepackage{tikz}

\usepackage[utf8]{inputenc}
\usepackage[T1]{fontenc}
\usepackage{latexsym}

\usepackage{amssymb}
\usepackage{makecell}
\usepackage{diagbox}

\usepackage{graphicx}
\usepackage{tabularx}
\usepackage{array}

\usepackage{booktabs}

\newtheorem{theorem}{Theorem}
\newtheorem{definition}{Definition}
\newtheorem{proposition}{Proposition}
\newproof{proof}{Proof}

\newtheorem{remark}{Remark}
\newtheorem{corollary}{Corollary}
\journal{}

\begin{document}

\begin{frontmatter}

\title{Algebraic Approach to Directed Rough Sets}

\author{A Mani\corref{cor1}}
\ead{a.mani.cms@gmail.com}
\address{HBCSE, Tata Institute of Fundamental Research\\
9/1B, Jatin Bagchi Road\\
Kolkata (Calcutta)-700029, India\\
Homepage: \url{http://www.logicamani.in}\\
Orchid: \url{https://orcid.org/0000-0002-0880-1035}}

\author{S\'{}andor Radeleczki}
\ead{matradi@uni-miskolc.hu}
\address{Institute of Mathematics, University of Miskolc\\
3515 Miskolc-Ergyetemv\'{}aros\\
Miskolc, Hungary\\
Homepage: \url{https://www.uni-miskolc.hu/~matradi}}
\cortext[cor1]{Corresponding author}

\maketitle

%\runninghead{A Mani \& S. Radeleczki}{Directed Rough Sets}

\begin{abstract}
In relational approach to general rough sets, ideas of directed relations are supplemented with additional conditions for multiple algebraic approaches in this research paper. The relations are also specialized to representations of general parthood that are upper-directed, reflexive and antisymmetric for a better behaved groupoidal semantics over the set of roughly equivalent objects by the first author. Another distinct algebraic semantics over the set of approximations, and a new knowledge interpretation are also invented in this research by her. Because of minimal conditions imposed on the relations, neighborhood granulations are used in the construction of all approximations (granular and pointwise). Necessary and sufficient conditions for the lattice of local upper approximations to be completely distributive are proved by the second author. These results are related to formal concept analysis. Applications to student centered learning and decision making are also outlined.
\end{abstract}


\end{frontmatter}

\section{Introduction}

In relational approach to general rough sets various granular, pointwise or abstract approximations are defined, and rough objects of various kinds are studied \cite{am501,am240,ppm2,gc2018,pp2018,gcd2018}. These approximations may be derived from information tables or may be abstracted from data relating to human (or machine) reasoning. A \emph{general approximation space} is a pair of the form $S = \left\langle \underline{S}, R \right\rangle$ with $\underline{S}$ being a set and $R$ being a binary relation ($S$ and $\underline{S}$ \emph{will be used interchangeably throughout this paper}). Often, approximations of subsets of $\underline{S}$ are generated from these and studied at different levels of abstraction in theoretical approaches to rough sets. It is also of interest to understand ideas of closeness of other relations to the relation $R$ -- this includes the problem of computing reducts of a type. 

Parthood (part of) relations \cite{ham2017,rgac15,katk06,ur,lp2011} of different kinds play a major role in human reasoning over multiple perspectives. They may be between objects and properties, or collections of objects or properties, or between concepts. For example, one can assert that \emph{red is part of maroon} or that \emph{red is a substantial part of pink} or that \emph{redness is part of pinkness} -- a key feature of such relations is the connection with ontology \cite{ham2017,am3930}. 

Rough Y-systems and granular operator spaces, introduced and studied extensively by the first author \cite{am501,am9969,am9411,am240}, are essentially higher order abstract approaches in general rough sets in which the primitives are ideas of approximations, parthood, and granularity. In the literature on mereology \cite{av,vie,ur,rgac15,am3930,seibtj2015}, it is argued that most ideas of binary \emph{part of} relations in human reasoning are at least antisymmetric and reflexive.  \emph{A major reason for not requiring transitivity of the parthood relation is because of the functional reasons that lead to its failure} (see \cite{seibtj2015}), and to accommodate \emph{apparent parthood} \cite{am9969}. In the context of approximate reasoning interjected with subjective or pseudo-quantitative degrees, transitivity is again not common. The role of such parthoods in higher order approaches are distinctly different from theirs in lower order approaches -- specifically, general approximation spaces of the form $S$ mentioned above with $R$ being a parthood relation are also of interest. Given two concepts ($A$ and $B$ say), it often happens that there are concepts like $E$ of which $A$ and $B$ are part of. This is, loosely speaking, the idea of the parthood relation being up-directed. In approximate reasoning with vague objects or concepts, this property is more common than the existence of supremums (in a general sense). 

From a purely mathematical perspective, the property of up-directedness (also referred to as directedness) of partial orders and semilattice orders is widely used in literature, it has also been used in studying concepts of \emph{ideals of binary relations} (see \cite{jc1977,am9204}. But the groupoidal approach of \cite{ichlps2015,ichlweak2013} is not known in earlier work.

In this research general approximation spaces, in which the relation $R$ is an up-directed parthood relation, are studied in detail by the authors. It is also shown that the algebraic semantics of such spaces is very distinct from those in which $R$ is a directed or partial or quasi-order. More specifically, two of the algebraic models are groupoids with additional operations (these correspond to granular approximations), while the third is based on completely distributive lattices (this corresponds to mixed local approximations).

In the following section, some of the essential background is mentioned. In the third section, directed rough sets are introduced and basic results are proved. Illustrative examples are invented in the following section. Algebraic semantics on the power set and subsets thereof are explored in depth by the first author in the fifth section. In the sixth section, groupoidal semantics over quotients are investigated by the first author. Algebraic semantics of local approximations, connections with formal concept analysis and induced groupoids on subsets of the power set are explored in the following section by the second author. Subsequently knowledge interpretation over the three semantic approaches is discussed and an application to student-centred learning is invented in the next section. Further directions are provided in the ninth section.

\section{Some Background}

\subsection{Information Tables}
The concept of \emph{information} can also be defined in many different and non-equivalent ways. In the first author's view \emph{anything that alters or has the potential to alter a given context in a significant positive way is information}. In the contexts of general rough sets, the concept of information must have the following properties:
\begin{itemize}
\item {information must have the potential to alter supervenience relations in the contexts (A set of properties $Q$ supervene on another set of properties $T$ if there exists no two objects that differ on $Q$ without differing on $T$),}
\item {information must be formalizable and }
\item {information must generate concepts of roughly similar collections of properties or objects.}
\end{itemize}

The above can be read as a minimal set of desirable properties. In practice, additional assumptions are common in all approaches and the above is about a minimalism. This has been indicated to suggest that comparisons may work well when ontologies are justified.   

The concept of an information system or table is not essential for obtaining a granular operator space or higher order variants thereof. As explained in \cite{am501,am9222,am9969}, in human reasoning contexts it often happens that they arise from such tables.  

Information tables (also referred to as descriptive systems or knowledge representation system in the literature) are basically representations of structured data tables. Often these are referred to as \emph{information systems} in the rough set literature, while it refers to an integrated heterogeneous system that has components for collecting, storing and processing data in AI, computer science and ML. From a mathematical point of view, the latter can be described using heterogeneous partial algebraic systems. In rough set contexts, this generality has not been exploited as of this writing. It is therefore suggested in \cite{cd2017} to avoid plural meanings for the same term.

An \emph{information table} $\mathcal{I}$, is a relational system of the form \[\mathcal{I}\,=\, \left\langle \mathfrak{O},\, \mathbb{A},\, \{V_{a} :\, a\in \mathbb{A}\},\, \{f_{a} :\, a\in \mathbb{A}\}  \right\rangle \]
with $\mathfrak{O}$, $\mathbb{A}$ and $V_{a}$ being respectively sets of \emph{Objects}, \emph{Attributes} and \emph{Values} respectively.
$f_a \,:\, \mathfrak{O} \longmapsto \wp (V_{a})$ being the valuation map associated with attribute $a\in \mathbb{A}$. Values may also be denoted by the binary function $\nu : \mathbb{A} \times \mathfrak{O} \longmapsto \wp{(V)} $ defined by for any $a\in \mathbb{A}$ and $x\in \mathfrak{O}$, $\nu(a, x) = f_a (x)$.

An information table is \emph{deterministic} (or complete) if
\[(\forall a\in At)(\forall x\in \mathfrak{O}) f_a (x) \text{ is a singleton}.\] It is said to be \emph{indeterministic} (or incomplete) if it is not deterministic that is
\[(\exists a\in At)(\exists x\in \mathfrak{O}) f_a (x) \text{ is not a singleton}.\]

Relations may be derived from information tables by way of conditions of the following form: For $x,\, w\,\in\, \mathfrak{O} $ and $B\,\subseteq\, \mathbb{A} $, $ \sigma xw $ if and only if $(\mathbf{Q} a, b\in B)\, \Phi(\nu(a,\,x),\, \nu (b,\, w),) $ for some quantifier $\mathbf{Q}$ and formula $\Phi$. The relational system $S = \left\langle \underline{S}, \sigma \right\rangle$ (with $\underline{S} = \mathbb{A}$) is said to be a \emph{general approximation space}. 

This universal feature of the definition of relations in general approximation spaces do not hold always in human reasoning contexts.

In particular if $\sigma$ is defined by the condition Equation \ref{pawl}, then $\sigma$ is an equivalence relation and $S$ is referred to as an \emph{approximation space}.
\begin{equation}\label{pawl}
\sigma xw \text{ if and only if } (\forall a\in B)\, \nu(a,\,x)\,=\, \nu (a,\, w) 
\end{equation}

\emph{In this research, prefix or Polish notation is uniformly preferred for relations and functions defined on a set. So instances of a relation $\sigma$ are denoted by $\sigma a b$ instead of $a \sigma b$ or $(a, b) \in \sigma$. If-then relations (or logical implications) in a model are written in infix form with $\longrightarrow$.} In Equation \ref{pawl}, \emph{if and only if} is used because the definition is not done in an obvious model.

\subsection{Context of this Research}

This research is relevant to all theoretical approaches to rough sets including the contamination avoidance based axiomatic granular approach due to the first author  \cite{am240,am501,am6900,am9969,am9204,am9114,am9006,am9222}, modal approaches (for the pointwise approximations) \cite{pp2018,ppm2}, and other abstract approaches \cite{jpr,cc5,ie2011}. For additional clarifications on the context, readers may refer to the references suggested.

In fact, the specific approximation spaces studied in this paper can be used to generate a number of 
High granular operator spaces and variants thereof studied by the first author \cite{am9114,am9969,am9006,am501,am9222}. These will be taken up in a separate paper.

\subsection{Algebraic Concepts}

For basics of partial algebras, the reader is referred to \cite{bu,lj}.
\begin{definition}
A \emph{partial algebra} $P$ is a tuple of the form \[\left\langle\underline{P},\,f_{1},\,f_{2},\,\ldots ,\, f_{n}, (r_{1},\,\ldots ,\,r_{n} )\right\rangle\] with $\underline{P}$ being a set, $f_{i}$'s being partial function symbols of arity $r_{i}$. The interpretation of $f_{i}$ on the set $\underline{P}$ should be denoted by $f_{i}^{\underline{P}}$, but the superscript will be dropped in this paper as the application contexts are simple enough. If predicate symbols enter into the signature, then $P$ is termed a \emph{partial algebraic system}.   
\end{definition}

In this paragraph the terms are not interpreted. For two terms $s,\,t$, $s\,\stackrel{\omega}{=}\,t$ shall mean, if both sides are defined then the two terms are equal (the quantification is implicit). $\stackrel{\omega}{=}$ is the same as the existence equality (also written as $\stackrel{e}{=}$) in the present paper. $s\,\stackrel{\omega ^*}{=}\,t$ shall mean if either side is defined, then the other is and the two sides are equal (the quantification is implicit). Note that the latter equality can be defined in terms of the former as 
\[(s\,\stackrel{\omega}{=}\,s \, \longrightarrow \, s\,\stackrel{\omega}{=} t)\&\,(t\,\stackrel{\omega}{=}\,t \, \longrightarrow \, s\,\stackrel{\omega}{=} t) \]

Various kinds of morphisms can be defined between two partial algebras or partial algebraic systems of the same or even different types. If \[X\, =\, \left\langle\underline{X},\,f_{1},\,f_{2},\,\ldots ,\, f_{n} \right\rangle \text{ and } W\, =\, \left\langle\underline{W},\,g_{1},\,g_{2},\,\ldots ,\, g_{n} \right\rangle \] are two partial algebras of the same type, then a map $\varphi \, :\, X\, \longmapsto\, W$ is said to be a 

\begin{itemize}
\item {\emph{morphism} if for each $i$, \[(\forall (x_1,\, \ldots \, x_k)\,\in \, dom (f_i)) \varphi (f_{i}(x_1 , \ldots , \, x_k))\,=\,  g_i (\varphi(x_1),\, \ldots , \, \varphi (x_k)) \]}
\item {\emph{closed morphism}, if it is a morphism and the existence of\\ $g_{i} (\varphi(x_1),\, \ldots , \, \varphi (x_k))$ implies the existence of $f_{i}(x_1 , \ldots , \, x_k)$.}
\end{itemize}

Usually it is more convenient to work with closed morphisms. 

\subsubsection{Lattice Concepts}

The reader may refer to \cite{gra1998,gra2014,jj,dp2002} for lattice theoretical concepts. Some are stated below for convenience.

In a complete lattice $L$, an element $x\neq 0$ is said to be \emph{completely join-irreducible} if and only if \[(\forall K\subseteq L)(\bigvee K = x \longrightarrow (\exists z\in K) z = x)\]
The set of join-irreducible elements of $L$ will be denoted by $CJ(L)$. The lattice $L$ is said to be \emph{CJ-generated} or \emph{spatial} if and only if every element of $L$ is represented as a join of some elements of $CJ(L)$.

A lattice in which every descending chain is finite is said to satisfy the descending chain condition (DCC). In particular, if a complete lattice satisfies DCC, then it is necessarily spatial.   

\subsection{Groupoids and Binary Relations}

Under certain conditions, groupoidal operations can correspond to binary relations on a set. More generally, all binary relations can be read as partial groupoidal operations in a perspective (\cite{ichlps2015}) and therefore all general approximation spaces can be transformed into partial groupoids. The connections will be explored by the first author in a forthcoming paper. In this subsection known results for groupoids are stated for convenience.

%The connections with $\lambda$-lattices are likely to be relevant \cite{sva}.   

Let $S=\left\langle \underline{S}, R \right\rangle $ be a relational system, define 
\[U_R (a, b) = \{x :\, Rax \,\&\, Rbx\} \]
$S$ is said to be \emph{up-directed} if and only if $U_R (a, b)$ is never empty. That is,
\begin{equation}
(\forall a, b) \, \neg U_R (a, b) = \emptyset  \tag{up-directed}
\end{equation}
\begin{definition}\label{updg}
If a relational system is up-directed, then it corresponds to a number of groupoids defined by 
\[(\forall a, b )\, ab = \left\lbrace  \begin{array}{ll}
 b & \text{if } Rab\\
 c & c\in U_R(a, b) \,\&\, \neg Rab\\
 \end{array} \right. \tag{updg}\]
\end{definition}
These are studied in \cite{icl2013}. The collection of groupoids satisfying the above condition will be denoted by $\mathfrak{B}(S)$ and an arbitrary element of it will be denoted by $\mathsf{B}(S)$. It may be noted that \emph{up-directed sets} (partially ordered sets that are up-directed) and related constructions are well-known in topology and algebra, but the specific association of up-directedness mentioned is new. 

\emph{Join directoids} \cite{jjq90} are groupoids of the form $S$ that admit of a partial order relation $\leq$ that satisfies $(\forall a, b)\, a, b \leq ab$ and if $\max\{a, b\}$ exists then $ab = \max \{a, b\}$.
Clearly the results of \cite{icl2013} may also be read as a severe generalization of known results for join directoids. It may also be noted that lambda lattices (that are commutative join and meet directoids) are related special cases (see \cite{sva,am105}).

\begin{theorem}[\cite{icl2013}] 
For a groupoid $A$, the following are equivalent
\begin{itemize}
\item {A up-directed reflexive relational system $S$ corresponds to $A$}
\item {$A$ satisfies the equations \[aa = a \, \&\, a(ab) = b(ab) = ab\]}
\end{itemize}
\end{theorem}

\begin{definition}
If $A$ is a groupoid, then two relational systems corresponding to it are $\Re (A) = \left\langle \underline{A}, R_A \right\rangle$ and $\Re^* (A) = \left\langle \underline{A}, R_A^* \right\rangle$ with 
\begin{align*}
R_A = \{(a, b):\, ab = b\}\\
R_A^* = \bigcup \{(a, ab),\,(b, ab)\}
\end{align*}
\end{definition}

\begin{theorem}[\cite{icl2013}]
\begin{itemize}
\item {If $A$ is a groupoid then $\Re^* (A)$ is up-directed.}
\item {If a groupoid $A\models a(ab) = b(ab) = ab$ then $\Re (A) = \Re^* (A)$.}
\item {If $S$ is an up-directed relational system then $\Re (\mathsf(B) (S)) = S$.}
\end{itemize}
\end{theorem}

\begin{theorem}[\cite{icl2013}]
 If $S = \left\langle \underline{S}, R \right\rangle$ is a up-directed relational system, then all of the following hold:
 \begin{itemize}
\item {$R$ is reflexive if and only if $\mathsf{B}(S) \models aa = a$.}
\item {$R$ is symmetric if and only if $\mathsf{B}(S) \models (ab)a = a$.}
\item {$R$ is transitive if and only if $\mathsf{B}(S) \models a((ab)c) = (ab)c$.}
\item {If $\mathsf{B}(S) \models ab = ba$ then $R$ is antisymmetric.}
\item {If $\mathsf{B}(S) \models (ab)a = ab$ then $R$ is antisymmetric.}
\item {If $\mathsf{B}(S) \models (ab)c = a(bc)$ then $R$ is transitive.}
\end{itemize}
\end{theorem}
 
Morphisms between up-directed relational systems are preserved by corresponding groupoids. A \emph{relational morphism}  (as in \cite{mal}) from a relational system 
 $S = \left\langle \underline{S}, R \right\rangle$ to another $K = \left\langle \underline{K},Q  \right\rangle$ is a map $f: S \longmapsto K$ that satisfies \[(\forall a, b)\, (Rab \, \longrightarrow Qf(a)f(b)).\]  $f$ is said to be \emph{strong} if it satisfies \[(\forall c, e\in Q)(\exists a, b\in S )\, Qf(a)f(b)\, \&\, f(a) = c, \&\, f(b) = e\]

\subsection{Approximation Spaces and Groupoids}\label{apprsp}

It should be noted that up-directedness is not essential for a relation to be represented by groupoidal operations. The following construction that differs in part from the above strategy can be used for partially ordered sets as well, and has been used by the first author in \cite{amdsc2016,am909} in the context of knowledge generated by approximation spaces. The method relates to earlier algebraic results including \cite{jjm,jj1978,kt1981,fjjm}. The groupoidal perspective can be extended for quasi ordered sets.

If $S = \left\langle \underline{S}, R \right\rangle$ is an approximation space, then define (for any $a, b\in S$)  
\begin{equation}
a\cdot b \,=\, \left\{
\begin{array}{ll}
a,  & \text{if } Rab \\
b, &  \text{if } \neg Rab 
\end{array}
\right. 
\end{equation}

Relative to this operation, the following theorem (see \cite{jjm}) holds:

\begin{theorem}\label{ab}
$\left\langle S,\, \cdot \right\rangle$ is a groupoid that satisfies the
following axioms (braces are omitted under the assumption that the binding is to the left,
e.g. '$abc$' is the same as '$(ab)c$'):
\begin{align*}
{x x = x} \tag{E1}\\
{x (a z) = (x a) (x z) } \tag{E2}\\
{x a x = x} \tag{E3}\\
{azxauz = auz } \tag{E4}\\
{u(azxa)z = uaz } \tag{E5}
\end{align*}
\end{theorem}

\begin{theorem}
The following are consequences of the defining equations of $\mathbb{E}_{0}$ (from \textbf{E1,E2,E3}): 
\begin{align*}
{x(ax) = x ;\;  x(xa) = xa ;\;  (xa)a = xa } \\
{x(xaz) = x(az) ;\;  (xz)(az) = xz ;\;  (xa)(zx) = xazx } \\
{xazxa = xa;\;  xazaz = xaz;\;  xcazaxa = xaza} \\
{(xazx)(za) = x(za) ;\;  x(az)a = xaza ;\;  (xaz)(ax) = (xza)(zx) ;\; xazxz = xzaz.} \\
{(\forall x)(ex=ea \longrightarrow x=a) \;\equiv \; (\forall x) xe = e } 
\end{align*}
\end{theorem}

\subsection{Meta Explanation of Terms}

This purpose of this list is to help with the terminology relating to general rough sets (and also high granular operator spaces \cite{am501,am9222}).

\begin{itemize}
\item {\textsf{Crisp Object}:  That which has been designated as \emph{crisp} or is an approximation of some other object.}
\item {\textsf{Vague Object}: That whose approximations do not coincide with the object or that which has been designated as a \emph{vague} object.}
\item {\textsf{Discernible Object}: That which is available for computations in a rough semantic domain (in a contamination avoidance perspective). }
\item {\textsf{Rough Object}: Many definitions and representations are possible relative to the context. From the representation point of view these are usually functions of definite or crisp objects.}
\item {\textsf{Definite Object}: An object that is invariant relative to an approximation process. In actual semantics a number of concepts of definiteness is possible. In some approaches, as in \cite{yzm2012,hmy2019}, these are taken as granules. Related theory has a direct connection with closure algebras and operators as indicated in \cite{am501}.}
\end{itemize}

\section{Up-Directed Rough sets: Basic Results}

\begin{definition}
In a general approximation space $S = \left\langle \underline{S}, R \right\rangle$ consider the following conditions:
\begin{align}
(\forall a, b )(\exists c) R ac \, \&\, Rbc  \tag{up-dir}\\
(\forall a) Raa   \tag{reflexivity}\\
(\forall a, b)(Rab \, \&\, Rba \longrightarrow a=b)   \tag{anti-sym}
\end{align}
If $S$ satisfies up-dir, then it will be said to be a \emph{upper directed approximation space}. If it satisfies all three conditions then it will be said to be a \emph{up-directed parthood space}.    
\end{definition}

In general, partial/quasi orders, and equivalences need not satisfy \textsf{up-dir}. When they do satisfy the condition, then the corresponding general approximation spaces will be referred to as \emph{up-directed general approximation spaces}.

The neighborhood granulations used for defining approximations are specified next.

\begin{definition}
For any element $a\in S$, the following neighborhoods are associated with it 
\begin{align}
[a] = \{x :\, Rxa\}   \tag{neighborhood}\\
[a]_i = \{x: \, Rax\}   \tag{inverse-neighborhood}\\
[a]_o = \{x: \, Rax \, \&\, Rxa \}   \tag{symmetric neighborhood}\\
\end{align}
A subset $A\subseteq S$ will be said to be \emph{nbd-closed} if and only if 
\[(\forall x\in A)\, [x] \subseteq A\] Let the set of all nbd-closed subsets of $S$ be $\mathcal{E}(S)$ 
\end{definition}

$[a]$ is the set of things that relate to $a$ and $[a]_i$ is the set of things that $a$ relates to. $[a]$, $[a]_i$ and $[a]_o$ are respectively denoted by  $R^{-1}(a)$, $R(a)$ and $(R\cap R^{-1})(a)$

\begin{definition}
For any subset $A\subseteq S$, the following approximations can be defined:
\begin{align}
A^{l}\,=\, \bigcup \{[a]:\, [a]\subseteq A\}   \tag{lower}\\
A^{l_i}\,=\, \bigcup \{[a]_i:\, [a]_i\subseteq A\}   \tag{i-lower}\\
A^{u}\,=\, \bigcup \{[a]:\, \exists z \in [a]\cap A\}   \tag{upper}\\
A^{u_i}\,=\, \bigcup \{[a]_i:\, \exists z \in [a]_i\cap A\}   \tag{i-upper}\\
A^{l_s}\,=\, \bigcup \{[a]_o:\, [a]_o\subseteq A\}   \tag{s-lower}\\
A^{u_s}\,=\, \bigcup \{[a]_o:\, \exists z \in [a]_o\cap A\}   \tag{s-upper}
\end{align}
\end{definition}

\begin{definition}
In the context of the previous definition, the pointwise and local approximations are defined as follows:

\begin{align}
\tag{Point-wise Upper} A^{u+} \, =\, \{ x \, :\, [x]\cap A \neq \emptyset \}.\\ 
\tag{Point-wise Lower} A^{l+} \, =\, \{x \,:\,[x]\subseteq A \}\\ 
\tag{Point-wise i-Upper} A^{ui+} \, =\, \{ x \, :\, [x]_i\cap A \neq \emptyset \}.\\ 
\tag{Point-wise i-Lower} A^{li+} \, =\, \{x \,:\,[x]_i\subseteq A \}\\ 
\tag{u-Triangle} A^{\triangle} \, =\, \{ x \, :\, Rax \, \& a\in A \}\\ %A^{\triangle}
\tag{l-Triangle} A^{\triangledown} \, =\, \{x \,:\,[x]_i\subseteq A\,\&\, x\in A \}\\ %A^{\triangledown}
\tag{ub-Triangle} A^{\blacktriangle} \, =\, \bigcup \{[x] :\, x\in A\}.\\ 
\tag{lb-Triangle} A^{\blacktriangledown} \, =\, \{x \,:\,[x]\subseteq A\,\&\, x\in A \}
\end{align}
\end{definition}

\begin{remark}
The *-triangle approximations are \emph{local} in the sense that they are defined relative to points (as opposed to subsets) in the set being approximated. It is shown below that while $\triangledown$ and $\blacktriangledown$ are pointwise approximation operators, $\triangle$ and $\blacktriangle$ are granular approximations because they can be represented as terms involving neighborhood granules and set operations alone. Because, neighborhoods and inverse-neighborhoods are used, the granular and pointwise approximations are inter related in a complex way.
\end{remark}

\begin{proposition}
In the above context, for any subset $A$,
\begin{align}
A^{\triangle} = \bigcup \{[x]_i :\, x\in A\} \subseteq A^{u+}\subseteq A^{u_i}\\
A^{\blacktriangle} = \bigcup \{[x]:\, x\in A\}\subseteq A^{u}
\end{align}
\end{proposition}

\begin{proposition}
In the above context, when $R$ is reflexive and $c$ is the complementation operation
\begin{align}
A^{l} = A^{\triangledown\triangle} \text{ and } A^{l_i} = A^{\blacktriangledown\blacktriangle}\\
A^{u} = A^{\triangle\blacktriangle} \text{ and } A^{u_i} = A^{\blacktriangle\triangle}\\
A^{\triangle} = \bigcup \{[x]_i :\, x\in A\} \subseteq A^{u+}\subseteq A^{u_i}\\
A\subseteq A^{\blacktriangle} = \bigcup \{[x]:\, x\in A\}\subseteq A^{u}\\
A^{\triangledown} \subseteq A\subseteq A^{\triangle}\,\&\, A^{\blacktriangledown} \subseteq A\subseteq A^{\blacktriangle}\\
A^{\triangle c} =  A^{c\triangledown}\,\&\, A^{\blacktriangle c} =  A^{c\blacktriangledown}
\end{align}
\end{proposition}

\begin{theorem}\label{luprop}
In a reflexive up-directed approximation space $S$, the following properties hold for elements of $\wp (S)$:
\begin{align}
(\forall a) a^{ll} = a^{l} \subseteq a   \tag{l-id}\\
(\forall a) a \subseteq a^{u} \subseteq a^{uu}   \tag{u-wid}\\
(\forall a) a^{l}\subseteq a^{lu} \subseteq a^u   \tag{lu-inc}\\
(\forall a, b)(a\subseteq b \longrightarrow a^l \subseteq b^l)   \tag{l-mo}\\
(\forall a, b)(a\subseteq b \longrightarrow a^u \subseteq b^u)   \tag{u-mo}\\
S^u = S = S^l \, \&\, \emptyset^l = \emptyset = \emptyset^u  \tag{bnd} \\
(\forall a, b)(a\cup b)^u = a^u \cup b^u   \tag{u-union}\\
(\forall a, b) a^l \cup b^l  \subseteq (a\cup b)^l  \tag{l-union}\\
(\forall a, b) (a\cap b = \emptyset \longrightarrow a^l \cup b^l  = (a\cup b)^l ) \tag{l-union0}\\
(\forall a, b) (a\cap b)^l \subseteq a^l \cap b^l   \tag{l-cap}\\
(\forall a, b)  (a\cap b)^u \subseteq a^u \cap b^u  \tag{u-cap}
\end{align}
\end{theorem}

\begin{proof}
\begin{description}
\item [l-id]{If $x\in a^{ll}$ then there exists a $b\in S$ such that $x\in [b] \subseteq a^l$. So $a^{ll} \subseteq a^{l}$. This proves \textsf{l-id}.}
\item [u-wid]{If $x\in a$ then because of reflexivity of $R$, $x\in [x]\subseteq a^u$. So \textsf{u-wid} holds.}
\item [lu-inc]{The proof of u-wid carries over to that of lu-inc because $a^l\subseteq a^u$ implies $a^{lu} \subseteq a^u$.}
\item [l-mo]{If $a\subseteq b$ then $(\forall x\in S)([x]\subseteq a \longrightarrow [x] \subseteq b)$. This ensures that $a^l \subseteq a \subseteq b $ and $b^l \subseteq b$ and $a^l \subseteq b^l$. }
\item [u-mo]{The proof is similar to that of l-mo.}
\item [bnd]{Follows from $(\forall x\in S) \, x\in [x]$.}
\item [u-union]{
\begin{itemize}
\item {If $x\in (a\cup b )^u $, then $(\exists z\in S)\, x\in [z]\,\&\, [z]\cap (a\cup b) \neq \emptyset$}
\item {the latter condition is  $([z]\cap a)\cup ([z]\cap b) \neq \emptyset$}
\item {or $([z]\cap a)\neq \emptyset$ or $([z]\cap b) \neq \emptyset$. So $x\in a^u \cup b^u$.}
\end{itemize}
\begin{itemize}
\item {Conversely, if $h\in a^u \cup b^u$ then $(\exists z\in S)\, h\in [z]\,\&\, ([z]\cap a) \neq \emptyset \,\vee \,([z]\cap b) \neq \emptyset $ }
\item {the latter condition is  $([z]\cap a)\cup ([z]\cap b) \neq \emptyset$}
\item {So $h\in (a\cup b)^u$}
\end{itemize}}
\item [l-union]{
\begin{itemize}
\item {$x\in a^l \cup b^l$ $\Leftrightarrow$ $x\in a^l$ or $x\in b^l$}
\item {$\Leftrightarrow (\exists z, h \in S)\, x\in [z]\subseteq a\, \vee \, x\in [h]\subseteq b $}
\item {$\Leftrightarrow (\exists z, h \in S)\, x\in [z] \subseteq a \cup b $ and so $x\in (a\cup b )^l$.}
\end{itemize}
Examples for the failure of the converse inclusion are easy to construct.}
\item [l-union0]{
\begin{itemize}
\item {$x\in (a \cup b)^l$ then  $(\exists z \in S)\, x\in [z] \subseteq a \cup b $}
\item {So $(\exists z \in S)\, x\in [z]\subseteq a\, \text{ Xor } \, x\in [z]\subseteq b $}
\item {So $x\in a^l$ xor $x\in b^l$, which implies $x\in a^l \cup b^l$.}
\end{itemize}}
\item [l-cap]{$x\in (a\cap b)^l$
\begin{itemize}
\item {if and only if $(\exists z\in S) x\in [z] \subseteq a\cap b$}
\item {if and only if $(\exists z\in S) x\in [z] \subseteq a \,\&\,[z] \subseteq a $}
\item {implies $x\in a^l$ and $x\in b^l$}
\end{itemize}
\begin{itemize}
\item {To see possible reasons for the failure of the converse, let $x\in a^l$ and $x\in b^l$}
\item {then $(\exists z_1 \in S)\, x\in [z_1] \subseteq a $ and $(\exists z_2 \in S)\, x\in [z_2] \subseteq b $}
\item {so $x\in [z_1]\cap [z_2] \subseteq a\cap b$, but it can happen that $[z_1]\cap [z_2]$ is not of the form $z$ for some $z\in a\cap b$.}
\end{itemize}}
\end{description}
\end{proof}

\begin{remark}
The nature of failure of $a^l \cap b^l \subseteq (a\cap b)^l$ shown in the proof suggests that it can be fixed at a semantic level in many ways.
\end{remark}

\begin{theorem}\label{lup}
In a up-directed approximation space $S$, the following properties hold for elements of $\wp (S)$:
\begin{align*}
(\forall a) a^{ll} = a^{l} \subseteq a   \tag{l-id0}\\
(\forall a) a^{u} \subseteq a^{uu}   \tag{u-wid0}\\
(\forall a) a^{l}\subseteq a^{lu} \subseteq a^u   \tag{lu-inc}\\
(\forall a, b)(a\subseteq b \longrightarrow a^l \subseteq b^l)   \tag{l-mo}\\
(\forall a, b)(a\subseteq b \longrightarrow a^u \subseteq b^u)   \tag{u-mo}\\
S^l =S^u \subseteq S  \, \&\, \emptyset^l = \emptyset = \emptyset^u  \tag{bnd0}\\
(\forall a, b)(a\cup b)^u = a^u \cup b^u   \tag{u-union}\\
(\forall a, b) a^l \cup b^l  \subseteq (a\cup b)^l  \tag{l-union}\\
(\forall a, b) (a\cap b)^l \subseteq a^l \cap b^l   \tag{l-cap}\\
(\forall a, b)  (a\cap b)^u \subseteq a^u \cap b^u  \tag{u-cap}
\end{align*}
\end{theorem}

\begin{proof}
Most of the proof of Theorem \ref{luprop} carries over. Because of the absence of reflexivity, the weaker properties u-wid0, and bnd0 hold.
\end{proof}

\begin{remark}
It may be noted that the upper cone of a subset $A$ (that is the set $\{b: \, (\exists a, c \in A)\, Rab \,\&\, Rcb\}$) is contained in $A^u$.  
\end{remark}

\section{Illustrative Examples}

Abstract and practical examples are constructed in this section for illustrating various aspects of up-directed approximation spaces.

\subsection{Abstract Example}\label{absexample}

Let $\underline{S}$ be the set \[\underline{S} \,=\, \{a, b, c, e, f\} \text{ and let }\]
\[R \, =\, \{ac, ae, af, bc, bf, ca, cb, cf, ea, ef, fa, fb \}\] be a binary relation on it ($ac$ \textsf{means the ordered pair} $(a, c)$ and so on for other elements). In Figure \ref{udr}, the general approximation space $S = \left\langle \underline{S}, R \right\rangle$ is depicted. An arrow from $e$ to $f$ is drawn because $Ref$ holds. 

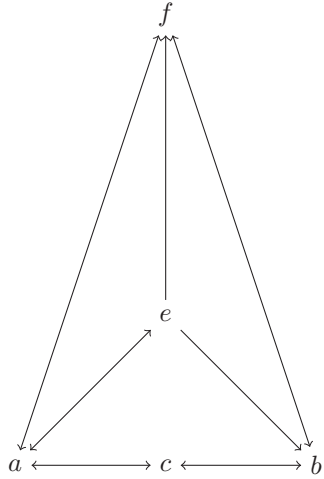
\begin{figure}[hbt]
\begin{center}
\begin{tikzpicture}[node distance=2cm, auto]
\node (F) {$f$};
\node (O) [below of=F] {};
\node (E) [below of=O] {$e$};
\node (C) [below of=E] {$c$};
\node (A) [left of=C] {$a$};
\node (B) [right of=C] {$b$};
\draw[->,font=\scriptsize] (E) to node {}(F);
\draw[<->,font=\scriptsize] (F) to node {}(A);
\draw[<->,font=\scriptsize] (E) to node {}(A);
\draw[<->,font=\scriptsize] (B) to node {}(F);
\draw[->,font=\scriptsize] (E) to node {}(B);
\draw[<->,font=\scriptsize] (A) to node {}(C);
\draw[<->,font=\scriptsize] (B) to node {}(C);
\end{tikzpicture}
\caption{Up-Directed Relation $R$}
\label{udr}
\end{center}
\end{figure}

The up-directed approximation space $S = \left\langle \underline{S}, R \right\rangle$ is irreflexive and $R$ is not antisymmetric. The antisymmetric completion $R^+$ of $R$ coincides with its reflexive completion and is defined by 
\[R^+ =  R \cup \{aa, bb, cc, ee, ff \}\]

The groupoid corresponding to $S$ is given by Table \ref{grps}

\begin{table}[h]
\centering
\begin{tabular}{|c|c c c c c|}
\hline
 & a & b & c & e & f \\
\hline
a & e & c & c & e & f \\
b & e & c & c & e & f \\ 
c & a & b & f & f & f \\
e & a & b & f & f & f \\
f & a & b & a & a & a \\
\hline
\end{tabular} 
\caption{A Groupoid of $S$}\label{grps}
\end{table}

The neighborhood granules determined by the elements of $S$ are as in Table \ref{nbdg}

\begin{table}[h]
\centering
\begin{tabular}{| c | c | c | c|}
\hline
$x $ & $ [x] $ & $ [x]_i $ & $ [x]_o $ \\
\hline
$a \, $ & $ \, \{c, e, f\}\, $ & $ \,\{c, e, f\}\, $ & $ \, \{c, e, f\}$ \\
\hline
$b\, $ & $ \, \{c, e, f\}\, $ & $ \,\{c, f\}\,$ & $ \, \{c,  f\}$ \\
\hline
$c \,$ & $ \,\{a, b \}\, $ & $ \,\{a, b, f \} \,$ & $ \, \{a, b\}$ \\
\hline
$e \,$ & $ \,\{a\}\, $ & $\, \{a, b, f\} \,$ & $\, \{a\}$ \\
\hline
$f \, $ & $ \,\{a, b, c, e \}\, $ & $ \,\{a, b\}\, $ & $\, \{a, b\}$ \\
\hline
\end{tabular}
\caption{Neighborhood Granules}\label{nbdg}
\end{table}

Since $\wp(S)$ has $32$ elements, approximations of specific subsets are alone considered next.

Let $A =\{e, c\} $, then its approximations are as below:
\begin{itemize}
\item {$A^l = \emptyset$ and $A^u = S$}
\item {$A^{\triangledown} =\emptyset $ and $A^{\triangle} = \{a, b, f\}$}
\item {$A^{\blacktriangledown} =\emptyset $ and $A^{\blacktriangle} = \{a, b\}$}
\item {$A^{l_i} =\emptyset = A^{l_o} $ and $A^{u_i} = \{c, e, f\} = A^{u_o}$}
\item {$A^{l+} =\emptyset $ and $A^{u+} = {a, b, f}$}
\item {$A^{l_i+} =\emptyset $ and $A^{u_i+} = \{a, b\}$}
\end{itemize}

\subsection{Reasoning about Vague Concepts}

Suppose a set $\underline{S}$ of concepts relating to a classroom lesson are given, and that some of these are vague. For any two concepts $a$ and $b$, assume that a concept $c$ that apparently contains the two exists -- this type of search for a $c$ amounts to taking decisions. Let this concept of apparent parthood be denoted by $R$. Depending on the context, the relation $R$ may be a up-directed, reflexive and antisymmetric relation. Thus $S= \left\langle \underline{S}, R \right\rangle$ may be a up-directed parthood space or definitely an up-directed space. 

\emph{Apparent parthood} relation has been considered by the first author in \cite{am9969} -- in general it is not antisymmetric. 

For two concepts $a$ and $b$, $ab = b$ may mean that $b$ fulfils the functions of $a$ in some sense (for example). If, on the other hand, $ab\in U_R(a, b)$ then there is a implicit reference to a choice function in the search for a concept that fulfils the role of both $a$ and $b$.   

For a concept $a$, the neighborhood $[a]$ is the set of concepts that are \emph{apparently part of} it, while $[a]_i$ is the set of concepts that it is apparently part of, and $[a]_o$ is the set of concepts that it is apparently part of and conversely. Obviously, when antisymmetry holds, the set $[a]_o$ will be a singleton. Note that these concepts have a directional character -- because of up-directedness of $R$. Each granule of the form $[a]$ may be associated with at least one element of $S$. Is $[a]$ \emph{determined by $a$}? The actual interpretation depends on the application context. In this case, it can be said the \emph{investigation of $a$ leads to the set} $[a]$. 

For a subset of concepts $A$, the lower approximation is an aggregation of directed granules that are included in $A$. It may also be read as the collection of \emph{relatively definite concepts} that are attainable from $A$ (using common sense methods or through common knowledge).

\section{Algebraic Semantics-1}\label{algs1}
%include power set and definites semantics only
In this section, possible semantics of the approximations $l$ and $u$ on their image set are investigated. From Theorem \ref{luprop} and Theorem \ref{lup}, it follows that a semantics over $\wp(S)$ without additional constructions is not justified because they do not distinguish between closely related general approximation spaces. 

\begin{definition}
On the set $(\wp(S))^u = \{x^u :\, x\in \wp(S)\} = S_u$, the following operations can be defined (apart from the induced $\cup$ operation):
\begin{align*}
a\wedge b \,=\, (a\cap b)^u \tag{iu1}\\
a\vee b \,=\, (a\cup b) \tag{iu2}\\
\bot = \emptyset   \tag{iu3}\\
\top = S^u  \tag{iu4}\\
\end{align*}
and the resulting algebra $S_u \,=\, \left\langle \underline{S_u}, \vee, \wedge, \cup, l, u, \bot , \top  \right\rangle$ will be called the \emph{algebra of upper approximations in a up-directed space} (UUA algebra). If $R$ is a up-directed parthood relation or a reflexive up-directed relation respectively, then it will be said to be a up-directed parthood algebra of upper approximations (UAP algebra) or a reflexive algebra of  upper approximations (UAR algebra) respectively.
\end{definition}
\begin{theorem}
The UUA, UAP and UAR algebras are well-defined, and an algebra of upper approximations satisfies all of the following:
\begin{align*}
(\forall a)\, a\vee a = a = a\vee \bot   \tag{idemp1}\\
(\forall a, b)\, a\wedge b = b\wedge a    \tag{comm2}\\
(\forall a, b)\, a\vee b = b\vee a    \tag{comm1}\\   
(\forall a, b, c)\, a\vee (b\vee c) = (a\vee b)\vee c  \tag{assoc1}\\
(\forall a)\, (a\wedge a)\vee a = a\wedge a = a^u   \tag{absfail}\\
(\forall a, b, c)\,(a\vee b = b \longrightarrow (a\wedge c)\vee (b\wedge c) = b\wedge c)  \tag{mo1}\\
\end{align*}
\end{theorem}

\begin{proof}
The lower approximation operation is redundant and so the algebras are well-defined. 
\begin{description}
\item[idemp1]{$a\vee a = a\cup a = a $.} 
\item[comm2]{$a\wedge b = (a\cap b)^u = (b\cap a)^u = b\wedge a$.}
\item[comm1]{Follows from definition.}
\item[assoc1]{Follows from associativity of set union.}
\item[absfail]{$a\wedge a = a^u$. So absorptivity fails in general.}
\end{description}
\end{proof}

Absorptivity can be improved by defining the operations differently. 

Let $S_{lu} = \{ x: \, x = a^l \text{ or } x= a^u \, \&\, a\in S\}$ 

\begin{definition}
 On $S_{lu}$, the following operations can be defined (apart from $l$ and $u$ by restriction):
 \begin{align*}
a\Cap b = (a\cap b )^l   \tag{Cap}\\
a \Cup b = (a\cup b)^u   \tag{Cup}\\
\bot = \emptyset   \tag{iu3}\\
\top = S^u  \tag{iu4}\\
\end{align*}
The resulting algebra $S_{lu} = \left\langle \underline{S_{lu}},\Cap , \Cup ,\cup, l, u, \bot , \top  \right\rangle$ will be called the \emph{algebra of approximations in a up-directed space} (UA algebra). If $R$ is a up-directed parthood relation or a reflexive up-directed relation respectively, then it will be said to be a up-directed parthood algebra of approximations (AP algebra) or a reflexive up-directed algebra of  upper approximations (AR algebra) respectively.
\end{definition}

\begin{theorem}
A AP algebra $S_{lu}$ satisfies all of the following:
\begin{align*}
(\forall a) a\Cap a = a \,\&\, (a\Cup a)\Cap a = a     \tag{idemp3}\\
(\forall a) a\Cup a = a^u \tag{quasi-idemp4}\\ 
(\forall a, b) a\Cap b = b\Cap a \,\&\, a\Cup b = b\Cup a   \tag{comm12}\\
(\forall a , b) a\Cap ( b \Cup a)  = a \tag{half-absorption}\\
(\forall a, b, c) a\Cup (b\Cup c) = (a \Cup b^u)\Cup c^u   \tag{quasi-assoc1}\\
(\forall a, b, c) (a\Cup (b\Cup c))\Cup ((a\Cup b)\Cup c) = ((a\Cup a)\Cup (b\Cup b))\Cup (c\Cup c \Cup c)    \tag{quasi-assoc0}
\end{align*}
\end{theorem}
\begin{proof}
\begin{description}
\item[idemp3]{ 
\begin{itemize}
 \item {$a\Cap a = (a\cap a)^l = a^l = a$}
 \item {$a\Cup a = a^u $ and $a^u \cap a = a$}
\end{itemize}}
\item[quasi-idemp4]{$a\Cup a = (a\cup a)^u = a^u$.} 
\item[comm12]{This follows from definition.}
\item[half-absorption]{
\begin{itemize}
 \item {$a\Cap (b\Cup a) = (a\cap (b\cup a)^u)^l = ((a\cap a^u)\cup (a\cap b^u))^l$}
 \item {$= (a \cup (a\cap b^u))^l = a^l = a$}
\end{itemize}} 
\item[quasi-assoc1]{
\begin{itemize}
 \item {$a\Cup (b\Cup c) = (a\cup (b\cup c)^u)^u = (a^u \cup b^{uu} \cup c^{uu})$}
 \item {$= (a \cup b^u))^u \cup c^{uu} = (a \Cup b^u)\Cup c^u$}
\end{itemize}}
\item[quasi-assoc0]{This can be proved by writing all terms in terms of $\cup$. In fact $(a\Cup (b\Cup c))\Cup ((a\Cup b)\Cup c) = a^{uuu}\cup b^{uuu} \cup c^{uuu}$. The expression on the right can be rewritten in terms of $\Cup$ by \textsf{quasi-idemp4}.}
\end{description}
\end{proof}

The above two theorems in conjunction with the properties of the approximations on the power set, suggest that it would be useful to enhance UA-, AP-, and AR-algebras with partial operations for defining an abstract semantics. 

\begin{definition}
A partial algebra of the form \[S_{lu}^* = \left\langle \underline{S_{lu}},\Cap , \Cup ,\cup,\sqcap, ^{\kappa}
, l, u, \bot , \top  \right\rangle\] will be called the \emph{algebra of approximations in a up-directed space} (UA partial algebra) whenever $S_{lu} = \left\langle \underline{S_{lu}},\Cap , \Cup ,\cup, l, u, \bot , \top  \right\rangle$ is a UA algebra and $\sqcap$ and $^{kappa}$ are defined as follows ($\cap$ and $^c$ being the intersection and complementation operations on $\wp (S)$):
\begin{equation}
(\forall a, b\in S_{lu} )\, a\sqcap b = \left\lbrace  \begin{array}{ll}
 a\cap b & \text{if } a\cap b \in S_{lu}\\
 \text{undefined} & \text{ otherwise}
 \end{array} \right. 
\end{equation}

\begin{equation}
(\forall a \in S_{lu} )\, a^{\kappa} = \left\lbrace  \begin{array}{ll}
 a^c & \text{if } a^c \in S_{lu}\\
 \text{undefined} & \text{ otherwise}
 \end{array} \right. 
\end{equation}

If $R$ is an up-directed parthood relation or a reflexive up-directed relation respectively, then it will be said to be a up-directed parthood partial algebra of approximations (AP partial algebra) or a reflexive algebra of upper approximations (AR partial algebra) respectively.
\end{definition}

\begin{theorem}
If $S$ is a up-directed approximation space, then its associated enhanced up-directed parthood partial algebra $S_{lu}^*= \left\langle \underline{S_{lu}},\Cap , \Cup ,\cup,\sqcap, ^{\kappa}
, l, u, \bot , \top  \right\rangle$ satisfies all of the following:
\begin{align*}
\left\langle \underline{S_{lu}},\Cap , \Cup ,\cup, l, u, \bot , \top  \right\rangle \text{ is a AP algebra}  \tag{app1}\\
(\forall a) \, a\sqcap a = a \,\&\, a\sqcap \bot = \bot \, \&\, a\sqcap \top = a    \tag{app2}\\
(\forall a, b, c) \, a\sqcap b \stackrel{\omega}{=} b \sqcap a\, \&\, a\sqcap (b\sqcap c) \stackrel{\omega}{=} (a\sqcap b)\sqcap c   \tag{app3}\\
a\sqcap a^u = a = a\sqcap a^l  \,\&\, a^{\kappa \kappa} \stackrel{\omega}{=} a  \tag{app4}\\
a\sqcap (b\cup c) \stackrel{\omega}{=} (a\sqcap b)\cup (a\sqcap c) \,\&\, a\cup (b\sqcap c) \stackrel{\omega}{=} (a\cup b)\sqcap (a\cup c)  \tag{app5}\\
(\forall a, b)\, (a\sqcap b)^{\kappa} \stackrel{\omega}{=} a^{\kappa}\cup b^{\kappa} \,\&\, (a\cup b)^{\kappa} \stackrel{\omega}{=} a^{\kappa}\sqcap b^{\kappa}  \tag{app6}
\end{align*}
\end{theorem}
\begin{proof}
The theorem follows from the previous theorems in this section.
\end{proof}

\section{Groupoidal Semantics}

\begin{definition}
In the powerset $\wp (S)$ generated by a upper directed approximation space $S$, the following operation can be defined (apart from the rough approximations and induced Boolean operations)
\begin{equation*}
(\forall A, B\in \wp(S))\, A\cdot B = \{ab : \,a\in A \, \&\, b\in B \}   \tag{g0}
\end{equation*}
The resulting algebra, $S^b = \left\langle \underline{\wp(S)}, \cdot, \cup, \cap, l, u, ^c, \bot, \top   \right\rangle$ of type $(2, 2, 2, 1, 1, 1, 0, 0 )$ will be called a \emph{basic power up-directed algebra} (\textsf{BP}-algebra). If $l$ and $u$ are replaced by $l_s$ and $u_s$, then the resulting algebra will be called a \emph{basic symmetric power up-directed algebra} (\textsf{BPS}-algebra)
\end{definition}

$a\subseteq b$ will be used as an abbreviation for $a\cup b = b$ in what follows.

\begin{theorem}
The algebra $\left\langle \underline{\wp(S)}, \cup, \cap, ^c, \emptyset, \wp(S)   \right\rangle$ is a Boolean algebra. Further, the following properties are satisfied by a BP-algebra $S^b$:
\begin{align*}
(\forall a, b, c)(a\cup b = b \longrightarrow ac \cup bc = bc)  \tag{order-comp}\\
(\forall a) \emptyset a = a\emptyset = \emptyset \,\&\, a S \subseteq S\, \&\, S a \subseteq S  \tag{bnd2}\\   
(\forall a, b, h)(a\cup b) h = (ah)\cup (bh) \,\&\, (a\cap b) h = (ah)\cap (bh)  \tag{comp2}\\
\text{ Conditions mentioned in eqn.\ref{luprop}. } \tag{lu-properties}   
\end{align*}
\end{theorem}

\begin{proof}
\begin{description}
\item [order-comp]{ If $x\in ac $, then it is of the form $ef$ with $e\in a$ and $f\in c$. By the premise, $e\in b$, so the conclusion follows.}
\item [comp2]{$x\in (a\cup b)h$ if and only if $x\in ah$ or $x\in bh$. Similarly for the second part.}
\end{description}
\end{proof}

\begin{remark}
Note that $x\in a^c h$ then $x$ is of the form $ef$ with $e\in a^c$ and $f\in h$, but $ef$ may be in $ah$ or $(ah)^c$. So, in general, $a^c h\neq (ah)^c$. 
\end{remark}

\subsection{Meaning of the Groupoidal Operation}

In the first author's opinion, the groupoid operation can be read in at least two ways. The operation obviously adds information to the general approximation space -- \emph{this addition can be read as a decision because it involves choice among alternatives}. In fact, the collection of all possible groupoidal operations can be used to generate a decision space. As such this aspect can be investigated in the given form or by taking the exact region to which the result of the operation belongs relatively. For the latter perspective, the groupoidal operation over $\wp (S)$ can be read as a combination of operations that are relatively better behaved relative to the approximations, aggregation and commonality operations. This permits easier interpretation, and semantics.

\begin{definition}
For any $A, B\in \wp(S)$, the following operations can be defined:
\begin{equation*}
n(A, B) \,=\, \{b:\, (\exists a\in A \exists b\in B)\, ab =b\} \tag{normal}
\end{equation*}
\begin{equation*}
o_1(A, B) \,=\, \{c:\, (\exists a\in A \exists b\in B)\, ab =c\in U_R(a, b)\setminus A \} \tag{outer-1}
\end{equation*}
\begin{equation*}
o_2(A, B) \,=\, \{c:\, (\exists a\in A \exists b\in B)\, ab =c\in U_R(a, b)\setminus B \} \tag{outer-2}
\end{equation*}
\begin{equation*}
i_1(A, B) \,=\, \{c:\, (\exists a\in A \exists b\in B)\, ab =c\in U_R(a, b)\cap A \} \tag{inner-1}
\end{equation*}
\begin{equation*}
i_2(A, B) \,=\, \{c:\, (\exists a\in A \exists b\in B)\, ab =c\in U_R(a, b)\cap B \} \tag{inner-2}
\end{equation*}
\begin{equation*}
 o(A, B) = o_1(A, B)\cap o_2(A, B) \tag{outer}
\end{equation*}
\end{definition}

In the above definition, the global groupoid operation has been split into multiple operations based on the relative values assumed. For any two sets $A, B\in \wp(S)$, 
\begin{itemize}
\item {$n(A, B)$ is the set of things in $B$ that have some part or approximate part in $A$,}
\item {$o_1(A, B$ is the set of things in the outer core determined by elements of $A\times B$ that are not in $A$,}
\item {$o_2(A, B$ is the set of things in the outer core determined by elements of $A\times B$ that are not in $B$,}
\item {$i_1(A, B$ is the set of things determined by elements of $A\times B$ that are in $A$,}
\item {$i_2(A, B$ is the set of things determined by elements of $A\times B$ that are in $B$, and}
\item {$o(A, B$ is the set of things determined by elements of $A\times B$ that are not in $A$ or $B$.}
\end{itemize}

These can also be read as generalizations of natural concepts of g-ideals in the context that can be defined as follows:
\begin{definition}
A subset $A$ of the  groupoid ${S =\left\langle\underline{S}, \cdot \right\rangle}$ is an \emph{g-ideal}  if and only if  $A$ is a subgroupoid and
\begin{equation}
a b =b \,\&\, b\in A \longrightarrow a\in A 
\end{equation}
A subset $B$ of the  groupoid $S$ is a \emph{g-filter}  if and only if  $B$ is a subgroupoid and
\begin{equation}
a b = b \,\&\, a\in A \longrightarrow b\in B 
\end{equation}
\end{definition}

The \emph{g-ideal generated by a subset} $A$ will be the smallest g-ideal $\mathbf{I}(A)$ containing the subgroupoid $Sg(A)$ generated by $A$.  If $A$ is a singleton, then the g-ideal will be said to be  \emph{principal}. The set of all g-ideals (resp principal, finitely generated) on $S$ will be denoted by $\mathcal{I}(S)$ (resp. $\mathcal{I}_p(S)$, $\mathcal{I}_f (S)$).

\begin{definition}
If $S$ is an up-directed parthood space, then the algebra \[S^{\sharp} = \left\langle \underline{\wp(S)}, n,  i_1,  i_2,  o_1,  o_2, o, \cup, \cap, l, u, ^c, \emptyset, \wp(S) \right\rangle\] defined above will be referred to as the \emph{expanded up-directed parthood\\ groupoidal Boolean algebra} (EUPGB) 
\end{definition}

\begin{theorem}
In the context of a EUPGB $S^{\sharp}$, all of the following hold (for any $A, B\in S^{\sharp}$):
\begin{align}
n(A, B) \subseteq B   \tag{n}\\
o_1(A, B) \subseteq A^c   \tag{o1}\\
o_2(A, B) \subseteq B^c   \tag{o2}\\
i_1(A, B) \subseteq A   \tag{i1}\\
i_2(A, B) \subseteq B   \tag{i2}\\
o(A, B) \subseteq (A\cup B)^c \tag{o}    
\end{align}
\end{theorem}

\begin{proof}
\begin{itemize}
\item {For any $a\in A$ and $b\in B$, $ab=b$ yields $ab\in B$.}
\item {For any $a\in A$ and $b\in B$, $ab=c\in U_R(a, b)\setminus A$ yields $ab\in A^c$. So o1 follows.}
\item {Note that if $a\in A$, $b\in B$, and $ab\in U_R(a, b)$ then it is possible that $ab\in B$. $o_2$ ensures that this does not happen. }
\item {Other parts can be verified from definition.}
\end{itemize}
\end{proof}

\begin{corollary}\label{subsetg}
If $B\subseteq A$ in the context of the previous theorem then
\begin{align}
n(A, B) = B   \tag{1}\\
i_2(A, B) \subseteq i_1(A, B) \subseteq A   \tag{2}\\
o_1(A, B)\subseteq o_2(A, B)   \tag{3}\\
AB = B\cup i_1(A, B)\cup o_2(A, B)   \tag{summary}
\end{align}
\end{corollary}

\begin{remark}
 Clearly the operations $n, \, i_1, \, i_2, \, o_1, \, 0_2$ and $o$ are better behaved than the groupoid operation $\cdot$.
\end{remark}

From the above considerations, it can also be deduced that

\begin{proposition}
In a EUPGB algebra $S^{\sharp}$, for any $a, b\in S$
\begin{itemize}
\item {$n([a],[b]) \subseteq [b]$}
\item {$i_1([a],[b])\subseteq [a]$}
\item {$i_2([a],[b])\subseteq [b]$}
\item {$o_1([a],[b]) \subseteq [a]^c$ and}
\item {$o_2([a],[b]) \subseteq [b]^c$}
\end{itemize}
\end{proposition}

\subsection{Rough Equalities and Inequalities}

Particular rough equalities of natural interest are defined next.

\begin{definition}
For any $a, b\in S^*$, let 
\begin{align*}
a\approx b \text{ if and only if } a^l = b^l \, \&\, a^u = b^u   \tag{standard req}\\
a\approx_l b \text{ if and only if } a^l = b^l \tag{l-standard req}\\
a\approx_u b \text{ if and only if }  a^u = b^u   \tag{u-standard req}
\end{align*}
$\approx$, $\approx_s$, $\approx_l$, $\approx_u$ and $\tau$ will respectively be referred to as the standard, l-standard, and u-standard rough equalities respectively.
\end{definition}

Obviously, the relations $\approx$, $\approx_s$, $\approx_l$, and $\approx_u$ are equivalences on $\wp (S)$.
The meaning of the above relations is closely connected with the following rough inequalities on $\wp(S)$:
\begin{definition}
In a EUGB $H$, the following relations are definable:
\begin{align}
a\sqsubseteq_l b \text{ if and only if } a^l \subseteq b^l   \tag{l-rough inequality}\\
a\sqsubseteq_u b \text{ if and only if } a^u \subseteq b^u   \tag{u-rough inequality}\\
a\sqsubseteq b \text{ if and only if } a\sqsubseteq_l b \, \&\, a\sqsubseteq_u b   \tag{rough inequality}
\end{align}
\end{definition}

\begin{proposition}
The relations  $\sqsubseteq_l$, $\sqsubseteq_u$ and $\sqsubseteq$ are quasi orders on the EUGB $H$. Moreover, they are partly compatible with the operations $\cup$ and $\cap$ in the following sense:
\begin{align*}
(\forall a, b, c)(a\sqsubseteq_l b \longrightarrow a\cap c \sqsubseteq_l b\cap c)    \tag{cs1}\\
(\forall a, b, c)(a\sqsubseteq_u b \longrightarrow a\cup c \sqsubseteq_u b\cup c)    \tag{cs2}\\
\end{align*}
\end{proposition}

\begin{proof}
It is obvious that $\sqsubseteq_l$ is reflexive. It is transitive because for any $a, b$ and $c$, $a^l \subseteq c^l$ follows from $a^l \subseteq b^l$ and $b^l \subseteq c^l$.

In general, antisymmetry does not hold because $a^l \subseteq b^l$ and $b^l \subseteq a^l$need not imply $a=b$.

The property \textsf{cs1} can be verified by considering the neighborhoods that may be included in $a$, $b$ and $c$, and observing that the neighborhoods included in $a\cap c$ must also be included in $b\cap c$.

The proof for $\sqsubseteq_u$ is analogous. For \textsf{cs2}, note that neighborhoods having nonempty intersection with $a\cup c$ must also have nonempty intersection with $b\cup c$. 
\end{proof}

In the context of a EUPGB, on the quotient $\wp (S)|\approx$, the following operations can be defined.

\begin{definition}\label{eugb}
 In the quotient $\wp(S)|\approx$ generated on a EUPGB \[\wp(S) = \left\langle \underline{\wp(S)}, n,  i_1,  i_2,  o_1,  o_2, o, \cup, \cap, l, u, ^c, \emptyset, \wp(S)   \right\rangle,\] the following operations can be defined 
\begin{equation*}
(\forall A, B\in \wp(S))\, \breve{\alpha}([A]_{\approx}, [B]_{\approx}) \,=\,  [\bigcup \{\alpha(F, H):\, F\in [A]_{\approx}\, \&\,H\in[B]_{\approx} \}]_{\approx}  
 \end{equation*}
 where $\alpha$ is any of $\cdot, n, i_1, i_2, o_1, o_2, o, \cup$ and $\cap  $. Further,
\begin{equation*}
(\forall A, B\in \wp(S))\, [A]_{\approx} \circledast [A]_{\approx}\,=\,  \left[\bigcap \{F\cap H:\, F\in [A]_{\approx} \,\&\, H\in [B]_{\approx} \}\right]_{\approx}  
 \end{equation*}
\begin{equation*}
(\forall A\in \wp(S))\, \neg([A]_{\approx}) \,=\,  \left[\bigcup \{F^c\in [A]_{\approx} \}\right]_{\approx}  
 \end{equation*}
 \begin{equation*}
(\forall A \in \wp(S))\, L([A]_{\approx}) \,=\,  \left[\bigcup \{F^l:\, F\in [A]_{\approx} \} \right]_{\approx} 
 \end{equation*}
 \begin{equation*}
(\forall A \in \wp(S))\, U([A]_{\approx}) \,=\,  \left[\bigcup \{F^u :\, F\in [A]_{\approx}\}\right]_{\approx}
 \end{equation*}
 
 In addition, the $0$-ary operations $\bot$ and $\top$ can be defined as $[\emptyset]_\approx$ and $[\wp(S)]_\approx$ respectively.
 
The algebra \[Z = \left\langle \underline{\wp(S)}|\approx, \breve{n},  \breve{i}_1,  \breve{i}_2,  \breve{o}_1,  \breve{o}_2, \breve{o}, \breve{\cup}, \breve{\cap}, L, U, \neg, \bot, \top   \right\rangle\] will be referred to as a \emph{up-directed rough parthood algebra} (RPA). 
\end{definition}

\begin{remark}
In the above definition, the $L$ and $U$ operations are not likely to behave as modal operators, and this is consistent with the semantic intent.  
\end{remark}

\begin{definition}
If $a$ is an element of $\wp(S)|\approx$ then it can also be interpreted as a subset of $\wp(S)$, and its \emph{representative approximations} $a_l$ and $a_u$ are  
\[a_l \,=\, x^l \text{ for any } x\in a \text{ and }\]  \[a_u \,=\, x^u \text{ for any } x\in a \]
\end{definition}

\begin{theorem}
All operations in Def.\ref{eugb} are well defined.
\end{theorem}

\begin{proof}
 The definition of each operation over $\wp(S)|\approx$ is based on forming a set of \emph{equivalent} sets from a union and so is well defined.
\end{proof}

In the next theorem, key relations between representatives and operations on a RPA are established.

\begin{proposition}\label{repop}
If $x\in a\in \wp(S)|\approx$, then $x$ can be represented in the form $a_l \cup K$ subject to the condition $a_l \cup K^u = a_u$ and $K^l =\emptyset$. 
\end{proposition}
\begin{proof}
Let $a = \{A_1, \ldots, A_n\}$ for some integer $n\leq \infty$ with $A_i \in \wp(S)$, then 
\[A_i = a_l \cup K_i \text{ for some set } K_i\]
Because $a_l^l \cup K_i^l \subseteq A_i^l = a_l$, it can be assumed that $K_i^l = \emptyset$ and that $K_i\cap a_l = \emptyset$.

In this situation, $A_i^u = (a_l \cup K_i)^u = a_l^u \cup K_i^u  = a_u$.
\end{proof}

\begin{theorem}
In the context of Prop. \ref{repop},  all of the following hold:
\begin{align*}
(\forall a) a_u \subseteq (Ua)_u    \tag{Uu}\\
(\forall a)a_l  = (La)_l \subseteq (Ua)_l   \tag{Ll}\\
(\forall a, b) (a\breve{\cup} b)_u = a_u \cup b_u   \tag{ujoins}\\
(\forall a) (\neg a)_u \subseteq a_l^{cu}   \tag{uc}\\
\end{align*}
\end{theorem}
\begin{proof}
\begin{description}
\item[Uu]{Using the representation of $a$ in Prop. \ref{repop}, it follows that $a_u = (a_l \cup K_i )^u = a_l ^u \cup K_i^u $ for any $i$, while \[Ua = [\bigcup \{a_l \cup b\, : \, b\cap K_i\neq \emptyset \text{ or } b\cap a\neq \emptyset \,\&\, b\in \mathcal{G} \}]\] So $Ua\, =\, (a_l)^u \cup \bigcup K_i^u $. So $a_u \subseteq (Ua)_u$. }
\item[Ll]{In the same representation, $La = [(\bigcup A_i)^l] = [a_l \cup (\bigcup K_i)^l]$. So $a_l = (La)_l$. }
\item[ujoins]{Suppose \[a\breve{\cup}b  = \left[\bigcup \{X_i ; X_i \in a  \text{ or } X_i \in b \} \right]\]
The $X_i$ can be written as $A_i \cup B_i = a_l \cup K_i \cup b_l \cup J_i$. Further for each $i$ $(a_l \cup K_i)^u = a_u = a_l^u \cup K_i^u$ ans similarly $b_u = b_l^u \cup J_i^u$. Using these for substituting $X_i$ results in $(a\breve{\cup} b)_u = a_u \cup b_u$.}
\item [uc]{Using the same strategy as in the proof of the previous properties, 
\begin{itemize}
 \item {$\neg a = \left[ \{\bigcup A_i^c \, :\, A_i \in a^ \}\right] =$}
 \item {$= \left[ \{\bigcup (a_l \cup K_i)^c \}\right] = \left[ a_l^c \cap (\bigcup K_i^c) \right].$}
 \item {So $(\neg a)_u \subseteq a_l^{cu}$}
\end{itemize}}
\end{description}
\end{proof}

\begin{remark}
This result suggests that a better (but relatively difficult) operation on the quotient $\wp(S)|\approx$ can be 
\begin{align*}
\mathfrak{L} a = \left[\{\bigcup \{ A : A\in a \}\}^l\right]_{\approx}  \tag{bL}\\
\mathfrak{U} a = \left[\{\bigcup \{ A : A\in a \}\}^u\right]_{\approx}  \tag{bU}
\end{align*}

It can be checked that while $\mathfrak{L} a $ is the same as $La$, but $(Ua)_l \subseteq (\mathfrak{U} a)_l$ and $(Ua)_u \subseteq (\mathfrak{U} a)_u$ in general.
\end{remark}

\begin{theorem}
If $Z$ is an RPA, then 
\begin{align*}
(\forall a \in Z)\,  a_l \sqsubseteq (La)_l \sqsubseteq (La)_u \tag{rep1}\\
(\forall a \in Z)\,  a_u \sqsubseteq (Ua)_u \tag{rep2}\\
(\forall a, b \in Z)\,  a_l \subseteq (a \breve{\cup} b)_l  \tag{rep3}\\
(\forall a, b \in Z)\,  a_u \subseteq (a \breve{\cup} b)_u  \tag{rep4}
\end{align*}
The converse of \textsf{rep1} holds if $Z$ is reflexive.
\end{theorem}
\begin{proof}
The proof depends on Thm \ref{luprop}, and Thm. \ref{lup}.
\begin{itemize}
\item {Suppose $a = \{b_1, b_2, \ldots , b_n \}$ with $b_i ^l = b_j^l = a_l$ for all $i,\, j$. By definition, $La = [\bigcup b_i^l]_{\approx}$. So $a_l \sqsubseteq (La)_l $. The converse holds if $Z$ is reflexive. }
\item {Suppose $a = \{b_1, b_2, \ldots , b_n \}$ with $b_i ^u = b_j^u = a_u$ for all $i,\, j$. By definition, $Ua = [\bigcup b_i^u]_{\approx}$. So $a_u \sqsubseteq (Ua)_u $.}
\item {Suppose $a = \{a_1, \ldots a_v \}$ and $b = \{b_1, b_2, \ldots , b_n \}$, then $a\breve{\cup} b$ is by definition equal to $[\bigcup \{a_i\cup b_j\}]_\approx$. Since $a_i \subseteq a_i\cap b_j $ for all $i$, \textsf{rep3} follows. }
\item {The proof of \textsf{rep4} is similar to that of \textsf{rep3}.}
\end{itemize}
\end{proof}

\begin{remark}
Note that the following can fail to hold in general:
\begin{align}
(\forall a, b \in Z)\,  (a \breve{\cap} b)_l \subseteq a_l \tag{rep5}\\
(\forall a, b \in Z)\,  (a \breve{\cap} b)_u \subseteq a_u \tag{rep6}
\end{align}
\end{remark}

\begin{theorem}
All of the following properties hold in a RPA $Z$: 
\begin{align*}
\breve{n}(a, a) = a   \tag{n-idemp}\\
a\breve{\cup} b \,=\, b\breve{\cup} a   \tag{join-comm}\\
a\breve{\cap} b \,=\, b\breve{\cap} a   \tag{meet-comm}\\
a\sqsubseteq a\breve{\cup} a    \tag{join-explosion}\\
a \sqsubseteq a \breve{\cap} a   \tag{meet-explosion}\\
a\circledast b = b\circledast a   \tag{star-comm}\\
a\sqsubseteq (a\breve{\cup} b)\breve{\cap} a  \tag{abs-fail}
\end{align*}
\end{theorem}

\begin{proof}
Let \[Z = \left\langle \underline{\wp(S)}|\approx, \breve{n},  \breve{i}_1,  \breve{i}_2,  \breve{o}_1,  \breve{o}_2, \breve{o}, \breve{\cup}, \breve{\cap}, L, U, \neg, \bot, \top   \right\rangle\]
\end{proof}

\subsection{Directoids}

As mentioned in the introduction, join directoids were introduced in \cite{jjq90}. A better equational way of defining these is as follows:
\begin{definition}
Directoids (join) are groupoids of the form $H = \left\langle \underline{H}, \cdot \right\rangle$ that satisfy the following conditions:
\begin{align*}
aa = a   \tag{dir1}\\
(ab)a = ab   \tag{dir2}\\
b(ab) = ab   \tag{dir3}\\
a((ab)c) = (ab)c   \tag{dir4}
\end{align*} 
\end{definition}

\begin{proposition}[\cite{ichlweak2013}]
A groupoid of the form $H = \left\langle \underline{H}, \cdot \right\rangle$ is a join directoid if and only if there exists a partial order $\leq$ on $H$ that satisfies
\begin{align}
(\forall a, b) a, b \leq ab   \tag{jd1}\\
(\forall a, b) (a\leq b \longrightarrow ab = ba = b)  \tag{jd2}
\end{align}
\end{proposition}

So it follows that a up-directed partially ordered set can be written as a groupoid and the groupoid in turn determines the partial order uniquely.

From the proposition it follows that 

\begin{theorem}
When an up-directed parthood space is also transitive, then a join directoid operation is definable on it (as per Equation \ref{updg}). 
\end{theorem}

\section{Algebraic Semantics of Local Approximations}

\begin{definition}
 If $\wp({S})^\triangle = \{A^\triangle:\, A\in \wp(S)\} $, then let 
 \begin{align}
 (\forall A, B\in \wp(S)^{\triangle})\, A\curlyvee B = A\cup B  \tag{T1}\\
 (\forall A, B\in \wp(S)^{\triangle})\, A\curlywedge B = (A\cap B)^{\blacktriangledown\triangle} \tag{T2}\\
 (\forall \{A_j\}_{j\in J} \in \wp(S)^{\triangle})\, \curlyvee_{j\in J} A_j = \bigcup A_j  \tag{ET1}\\
 (\forall \{A_j\}_{j\in J}\in \wp(S)^{\triangle})\, \curlywedge_{j\in J} A_j = (\bigcap_j A_j)^{\blacktriangledown\triangle} \tag{T2}
 \end{align}
\end{definition}

\begin{theorem}\label{tr1}
The algebra $\left\langle \underline{\wp(S)^{\triangle}}, \curlyvee, \curlywedge  \right\rangle $ is a complete bounded lattice. The corresponding lattice order on the algebra is $\subseteq$ (induced from set inclusion on $\wp(S)$).
\end{theorem}
\begin{proof}
If $A = X^{\triangle}$  for some $X$, then $A = \bigcup_{x\in X} [x]_i$. For arbitrary collections $\{A_j\}_J$ in $\wp(S)^\triangle$, it is easy to see that 
\[(\forall B)(\& A_j\subseteq B \longrightarrow \, \bigcup A_j \subseteq B)\]
This ensures that the union is a complete join semilattice operation and $B = Z^{\triangle}$  for some $Z$, then $A = \bigcup_{x\in Z} [x]_i$

\begin{align*}
\curlywedge_{j\in J}A_j = \bigcup \{Z\,:\, Z\in \wp(S)^{\triangle}\, \&\, Z\subseteq \bigcap_{j\in J}A_j\} = \bigcup \{\cup[x]_i:\, \cup [x]_i\subseteq \bigcap_{j\in J}\, A_j\} =\\
\bigcup \{[x]_i:\,  [x]_i\subseteq \bigcap_{j\in J}\, A_j\} = (\bigcap_{j\in J} A_j )^{\blacktriangledown\triangle}
\end{align*}
\end{proof}

\begin{theorem}\label{tr2} $ \left\langle  \wp (S)^{\triangledown
},\subseteq \right\rangle  $ is dually isomorphic to 
$ \left\langle  \wp (S)^{\triangle},\subseteq \right\rangle  $ as a complete lattice.
\end{theorem}

\begin{proof}
Define a map $f: \wp(S)^{\triangle}\longmapsto \wp (S)^{\triangledown}$ according to 
\[(\forall Z \in \wp(S)^{\triangle})\, f(Z) = Z^c \]
\begin{itemize}
\item {Since any $Z\in\wp (S)^{\triangle}$ has the form $Z=X^{\triangle}$ for some $X\subseteq S$, 
so $f(Z)=X^{\triangle c}=X^{c\triangledown}\in\wp (S)^{\triangledown}$. }
\item {This ensures that the map $f$ is well defined.}
\item {For any $Z_{1},Z_{2}\in\wp(S)^{\triangle}$, 
\[Z_{1}\subseteq Z_{2}\leftrightarrow Z_{2}^{c}\subseteq Z_{1}^{c}\leftrightarrow f(Z_{1})\supseteq f(Z_{2}).\]  From this it follows that $f$ is a dual order-isomorphism.}
\end{itemize}
Hence in view of Theorem \ref{tr1},
$ \left\langle  \wp (S)^{\triangledown},\subseteq \right\rangle  $ is also a complete lattice.
\end{proof}

\begin{definition}
On the image $\wp (S)^{\blacktriangle}=\{X^{\blacktriangle}: X\subseteq
S\}$ of $\blacktriangle$, the induced relation $\subseteq$ can be associated with the following operations:
\begin{equation}
\underset{i\in I}{\curlyvee ^*}A_{i}=\underset{i\in I}{\ \bigcup}A_{i} \tag{bt-join}
\end{equation}
\begin{equation}
\underset{i\in I}{\curlywedge^*}A_{i}= ( \underset{i\in I}{\bigcap}%
A_{i}) ^{\triangledown\blacktriangle} \tag{bt-meet}
\end{equation}
\end{definition}

Note that relation of $\wp (S)^{\blacktriangle}$ to $R^{-1}$ corresponds to the relation of $\wp (S)^{\triangle}$ with $R$.

\begin{theorem}\label{tr3}
$ \left\langle  \wp (S)^{\blacktriangle},\subseteq \right\rangle  $ and $ \left\langle  \wp (S)^{\blacktriangledown},\subseteq \right\rangle $ are dually isomorphic complete lattices.
\end{theorem}
\begin{proof}
The proof is analogous to that of Theorem \ref{tr2}.
\end{proof}

Definable sets in rough sets can be described in different ways. From a lattice-theoretical perspective, it is of interest to see if the set of lower or upper definable or at least the set of lower and upper approximations form distributive lattices. In this section, it is shown that the algebras formed by the set of approximations $\wp (S)^{\triangle}$, $\wp (S)^{\triangledown}$, $\wp (S)^{\blacktriangle}$, and $\wp (S)^{\blacktriangledown}$ are completely distributive lattices. It may be noted that the second author has studied these sets from a similar perspective in the context of approximations generated by tolerance relations in \cite{jjsr2017}.

In view of Theorem \ref{tr3}, this condition is equivalent to the condition that the concept
lattice $\mathcal{L}(S,S,I)$ is (completely) distributive. In \cite{gw99} several
conditions equivalent to the complete distributivity of $\mathcal{L}(S,S,I)$
are formulated. For instance, the following was established:

\begin{theorem}[{\cite{gw99}: Thm.40}]\label{th40} 
A concept lattice $\mathcal{L}(G,M,I)$ is completely distributive if and only if for any
object attribute pair $(g,m)\notin I$ there exists an object $h\in G$ and an attribute $n\in M$ with $(g,n)\notin I$, $(h,m)\notin I$ and such that $h\in\{k\}^{II}$, for any $k\in G\diagdown\{n\}^{I}$.
\end{theorem}

As an immediate consequence, in case of the concept lattice $\mathcal{L}(S,S,I)$ and the lattice $\wp (S)^{\triangle}$ we can formulate the following:

\begin{theorem}\label{tr23} The lattice $ \left\langle  \wp(S)^{\triangle},\subseteq \right\rangle  $ is completely distributive if and only if for any $a,b\in S$ satisfying $Rab$ there exist some elements $n,h\in S$
satisfying  $Ran \, \&\, Rhb $ and such that for any $x\in S$ satisfying $(Rxn$ we have $[h]_i \subseteq [x]_i$. 
That is 
\[(\forall a, b) Rab \longrightarrow (\exists n, h)(\forall x) Ran\,\&\, Rhb \, \&\, [h]_i \subseteq [x]_i\]
\end{theorem}
\begin{proof}
In view of Theorem \ref{tr1}, $ \left\langle  \wp(S)^{\triangle},\subseteq \right\rangle  $ is completely distributive if and only if the concept lattice $\mathcal{L}(S,S,I)$ is completely distributive. This is
equivalent to the condition formulated in Theorem \ref{tr2}. 
\begin{itemize}
\item{Let $a, b\in S$ and $Rab$.}
\item {In the context $\mathcal{L}(S,S,I)$, $Rab \leftrightarrow \neg Iab$. So the above
theorem applies with $g:=a$ and $m:=b$ and there exists $n,h\in S$ with $Ran\, \&\, Rhb$ and satisfying $h\in\{x\}^{II}$, for any $x\in S\diagdown\{n\}^{I}$.}
\item { As $S\diagdown\{n\}^{I}=\{s\in S:\neg Isn \}$, $x\in S\diagdown\{n\}^{I}$ means that $Rxn$. Since
$h\in\{x\}^{II}$ is equivalent to $[x]_i^{c}=\{x\}^{I}\subseteq\{h\}^{I}=[h]_i^{c}$, we deduce that $ Rxn$ implies $[h]_i\subseteq [x]_i$, for any $x\in S$.}
\end{itemize}

Therefore the condition in the present theorem is equivalent to the condition formulated in Theorem \ref{th40} and the conclusion follows.
\end{proof}

By using this theorem, two characterizations of the (complete) distributivity of $ \left\langle  \wp (S)^{\triangle},\subseteq \right\rangle  $ can be deduced. Also note that it is easy to check that any completely distributive
element of $ \left\langle  \wp (S)^{\triangle},\subseteq \right\rangle  $ has the form $[s]_i$ (for some $s\in S$) -- but the converse statement is not true in general.

\begin{theorem}\label{tr24}
If the lattice $ \left\langle \wp (S)^{\triangle},\subseteq \right\rangle $ is spatial, then the
following assertions are equivalent:
\begin{description}
\item [i]{The lattice $ \left\langle  \wp (S)^{\triangle},\subseteq \right\rangle  $ is completely distributive.}
\item [ii]{If $[s]_i$ is an arbitrary completely join-irreducible element of $ \left\langle  \wp (S)^{\triangle},\subseteq \right\rangle  $, then
\begin{equation}
[s]_i\nsubseteqq\bigcup\{[x]_i: [x]_i\nsupseteqq [s]_i\}  \tag{ei3}
\end{equation}}
\end{description}
\end{theorem}

\begin{proof}
[i] $\Rightarrow$ [ii]

\begin{itemize}
\item {Let $[s]_i$ be a completely join-irreducible element of $ \left\langle  \wp (S)^{\triangle},\subseteq
 \right\rangle  $.}
\item {Assume the contrary $[s]_i\subseteq\bigcup \{[x]_i: [x]_i\nsupseteqq [s]_i\}$.}
\item {Since the lattice $ \left\langle  \wp(S)^{\triangle},\cup,\wedge \right\rangle  $ is completely distributive, we have $[s]_i=[s]_i\wedge \bigcup\{[x]_i: [x]_i\nsupseteqq [s]_i\} =\bigcup\{[x]_i\wedge [s]_i: [x]_i\nsupseteqq [s]_i\}$.} 
\item {Since $[s]_i$ is a completely join-irreducible, we obtain $[s]_i=[s]_i\wedge [x]_i$,
i.e. $[s]_i\subseteq [x]_i$, for some $x\in S$ with $[x]_i\nsupseteqq [s]_i$ -- a
contradiction.}
\item {This proves the implication.}
\end{itemize}

[ii] $\Rightarrow$ [i]

\begin{itemize}
\item {Assume that (ii) holds, and let $Rab$ for some $a,b\in S$. Then $b\in [a]_i$.}
\item {If $[a]_i$ is completely join-irreducible, then in view of (ii), there exists an element $n\in [a]_i\setminus \bigcup\{[x]_i:\, [x]_i\nsupseteqq [a]_i\}$. If we set $h:=a$, then $Ran \,\&\, Rhb$.}
\item {For any $k\in S$ satisfying $Rkn$, $n\in [k]_i$ excludes the case $[k]_i\nsupseteqq [a]_i$, hence we obtain $[k]_i\supseteq [a]_i= [h]_i$.}
\item {Now suppose that $[a]_i$ is not completely join-irreducible. Then $b\in [a]_i=\bigcup\{[p]_i:\, [p]_i\in\ CJ(\wp (S)^{\triangle})\}$, and this yields $b\in [p]_i$ for some completely join-irreducible element $[p]_i)$ of
$\wp (S)^{\triangle}$ with $[p]_i\subseteq [a]_i$. Further $Rpb$ and 
\[[p]_i\nsubseteqq\bigcup\{[x]_i: [x]_i\nsupseteqq [p]_i\}\]}
\item {Therefore there exists an element $n\in [p]_i\setminus \bigcup\{[x]_i: [x]_i\nsupseteqq R(p)\} \subseteq [a]_i$ and hence we get $Ran$. Set $h:=p$. This yields $Rhb$.}
\item {For any $k\in S$ satisfying $Rkn$, $n\in [k]_i$ and $n\notin \bigcup\{[x]_i: [x]_i\nsupseteqq [p]_i\}$ exclude $[k]_i\nsupseteqq [p]_i$. Hence we obtain $[h]_i\subseteq [k]_i$.}
\end{itemize}
From this it follows that the lattice $ \left\langle  \wp (S)^{\triangle},\subseteq \right\rangle  $ is completely distributive.
\end{proof}

Replacing the relation $R$ with $R^{-1}$ in the above theorem we obtain:

\begin{theorem}\label{cor2.5}
If the lattice $ \left\langle  \wp(S)^{\blacktriangle},\subseteq \right\rangle$ is spatial, then the following
assertions are equivalent:
\begin{enumerate}
\item {The lattice $ \left\langle  \wp (S)^{\blacktriangle},\subseteq \right\rangle $ is completely distributive.}
\item {If $R^{-1}(s)$ is a completely join-irreducible element of $\left\langle  \wp (S)^{\blacktriangle},\subseteq \right\rangle$, then 
\begin{equation}
[s]\nsubseteqq\bigcup\{[x]:\, [x] \nsupseteqq [s] \}  \tag{ei4}
\end{equation}}
\end{enumerate}
\end{theorem} 

\begin{theorem}\label{cor2.6}
Let $R$ be a reflexive antisymmetric relation. Then
\begin{description}
\item [i]{$ \left\langle  \wp (S)^{\triangle},\subseteq \right\rangle$ is completely distributive if and only if $R$ is transitive.}
\item [ii]{$ \left\langle  \wp (S)^{\triangle},\subseteq \right\rangle  $ is completely distributive if and only if $ \left\langle \wp(S)^{\blacktriangle},\subseteq \right\rangle$ is completely distributive.}
\end{description}
\end{theorem}

\begin{proof}
\begin{itemize}
\item {If $R$ is a reflexive and antisymmetric relation, then in view of Theorem \ref{tr23}, $ \left\langle  \wp (S)^{\triangle},\subseteq \right\rangle  $ and $ \left\langle  \wp (S)^{\blacktriangle},\subseteq \right\rangle  $
are spatial lattices and for any $s\in S$, $[s]_i$ is a completely join-irreducible element of $\wp (S)^{\triangle}$, and $[s]$ is completely join-irreducible in $\wp (S)^{\blacktriangle}$.}
\item {Therefore, in view of Theorem \ref{tr24} and Theorem \ref{cor2.5}, the relations (ei3) and (ei4) are satisfied for all $s\in S$.}
\end{itemize}

\begin{description}
\item [i]{\begin{itemize}
\item {Suppose that $ \left\langle  \wp (S)^{\triangle},\subseteq  \right\rangle  $ is completely distributive, and let $Rus \,\&\, Rsv$.}
\item {Then $[s]_i\nsubseteqq\bigcup\{[x]_i: [x]_i\nsupseteqq [s]_i\}$, according to (ei3).}
\item {Let $n\in [s]_i\setminus \bigcup\{[x]_i: [x]_i\nsupseteqq [s]_i\}$, then $Rsn$ and $n\in [n]_i$ implies that $[s]_i\subseteq [n]_i$.}
\item {Since $s\in [s]_i$, we also get $Rns$. By the antisymmetry of $R$ we obtain $n=s$. Hence $s\in [s]_i\setminus \bigcup\{[x]_i: [x]_i\nsupseteqq [s]_i\}$.}
\item {Since $s\in [u]_i$, we have $[s]_i\subseteq [u]_i$. As $v\in [s]_i$ (by assumption), we obtain $v\in [u]_i
\,\&\, Ruv$. This proves the transitive property of $R$.}
\item {Conversely, if it is assumed that $R$ is transitive, then $R$ is a partial order, and by the result in \cite{jrv09} $\wp (S)^{\triangle}$ and $\wp(S)^{\blacktriangle}$ are completely distributive lattices.}
\end{itemize}}
\item [ii] {\begin{itemize}
\item {Assume that $ \left\langle  \wp (S)^{\triangle},\subseteq  \right\rangle  $ is completely distributive. Then, in view of (i) $R$ and $R^{-1}$ are partial orders.}
\item {Then by \cite{jrv09} $ \left\langle  \wp (S)^{\blacktriangle},\subseteq \right\rangle  $ is also completely distributive.}
\item {The proof of the converse implication is completely analogous.}
\end{itemize}}
\end{description}
\end{proof}

By applying the above definitions and Proposition 5, we obtain:

\begin{theorem}
Let $(S,R)$ be a directed relational system and $ \left\langle  B(S),\cdot \right\rangle  $ the corresponding
groupoid. Then $ \left\langle  \wp (S)^{\triangle},\subseteq \right\rangle  $ is completely distributive if and only if
\begin{multline}
(\forall a, b\in S)(ab= b \longrightarrow (\exists n, h\in S)(\forall k, s\in S)\\ \, an=n \,\&\, hb= b \, \&\, kn =n\, \& (hs = s \rightarrow ks =s) ) \tag{triagrp}
\end{multline}
\end{theorem}

\section{Knowledge Perspective}\label{seckno}

In a general approximation space $S$, if $R$ is an equivalence, a partial order or a quasi order, then it is also possible to associate other groupoidal operations (see \cite{amdsc2016,am501,am5019,am909}) on $S$. This is discussed in brief in Sec.\ref{apprsp}. But the associated operation is distinct from the one considered in this paper. 

General and classical rough sets have been associated with concepts of knowledge and studied from that perspective in a number of papers by the first author \cite{am9114,am9006,am9501,am9969,am99} and others\cite{zpb,zp6,ppm2,chp3,bgc12}. The basic idea in the context of classical approximation spaces \cite{zpb} is to associate definite objects with concepts and consequently the equivalence relation $R$ is  associated with knowledge. In more general situations, granularity has a bigger role to play, and knowledge is defined relative to granular axioms used and other desirable properties. Examples of such conditions are
\begin{description}
\item [GK1]{Individual granules are atomic units of knowledge.}
\item [GK2]{If collections of granules combine subject to a concept of mutual
independence, then the result would be a concept of knowledge. The 'result' may
be a single entity or a collection of granules depending on how one understands the
concept of \emph{fusion} in the underlying mereology.}
\item [GK3]{Maximal collections of granules subject to a concept of mutual independence are admissible concepts of knowledge.}
\item [GK4]{Parts common to subcollections of maximal collections may be interpreted as knowledge.}
\item [GK5]{All stable concepts of knowledge consistency should reduce to correspondences between granular components of knowledges. In particular, two relations $R_1$ and $R_2$ may be said to be \emph{consistent} if and only if the set of granules associated with the two general approximation spaces have bijective  correspondence. }
\end{description}
In \cite{am99} and \cite{dtl2017} choice operations over granules are involved. But they do not generate groupoid operations on the general approximation space itself. Neither do the granular knowledge axioms of the kind mentioned. All this means that the groupoid operation provides an additional layer of decision making that needs to integrated with existing work. A concrete practical example is considered next to illustrate key aspects of this.

\subsection{Applications to Student Centred Learning}

In student-centered learning students are put at the center of the learning process, and are encouraged to learn through active methods. Arguably, students become more responsible for their learning in such environments. In traditional teacher-centered classrooms, teachers have the role of instructors and are intended to function as the only source of knowledge. By contrast, teachers are typically intended to perform the role of facilitators in student-centered learning contexts. A number of best practices for teaching in such contexts \cite{jrp2016} have evolved over time. Teachers need to constantly improve their methods in such teaching contexts because that is part of the methodology.

Because of the open-ended aspect of the learning process, it is not expected that teachers have absolute control over the concepts learned. Students may themselves arrive at new methods of solution or define new concepts as part of the learning process. In this scenario it is of interest to suggest potential higher concepts that relate to the progress of the work in question. Teachers can possibly provide some initial suggestions and subsequently these can be worked upon by algorithms relying upon datasets of concepts for improved suggestions. From the perspective of this research this becomes the problem of construction of the best groupoid operations.

In more precise terms,
\begin{description}
\item [L1]{Let $A$ and $B$ be two concepts arrived at by the learner. The open-ended nature of the learning process means that a general rough set model of concepts must be adaptive or permit supervision.  }
\item [T1]{Teacher observes that concept $C$ among others contains $A$ and $B$ in some sense, and offers suggestions relating to the scenario.}
\item [S1]{Software aid for the learning context provides better suggestions based on \textsf{L1} and \textsf{T1} using a groupoidal decision model instead of the former alone. In general available strategies that can be used to arrive at suggestions based on \textsf{L1} alone are likely to be unintelligent.}
\end{description}

It may be noted that the impact of AI on enhancing classroom learning and learning in general has been very limited (see \cite{chos2018} and related references). In fact digital technology in the context of mathematics teaching has been stagnating because most of the effort has been on non-intelligent software that merely aid communication. There is no dearth of motivation for such work -- Often teachers do not have sufficient knowledge about the working of their students mind, have an excess of work load at hand and may be suffering from cognitive dissonances of specific types. 

In a forthcoming paper by the first author, the rough methodology suggested in this subsection is applied to specific practices such as opening of exercises in the context of mathematics teaching \cite{sko2011,milani2019}, use of explicit mathematical language \cite{usz2012}, and software for student expression \cite{alp2019,chos2018}.

\section{Further Directions and Remarks}

In this research 
\begin{itemize}
\item {the concepts of up-directed and up-directed parthood approximation\\ spaces are invented,}
\item {their potential role in weak decision making is illustrated, }
\item {algebraic semantics of sets of granular, nongranular and local approximations are invented and investigated in depth and shown to be nonequivalent, }
\item {algebraic semantics of roughly equivalent objects that involve additional groupoidal operations of decision making are invented and investigated,}
\item {their connection with knowledge and formal concept analysis are explored, and}
\item {possible applications to student centred learning is proposed.}
\end{itemize}
The results on connection with FCA supplement the work in \cite{jjs2014}. 

Parthood and apparent parthood relations have been the focus in higher order granular approaches to rough sets in a number of papers by the first author \cite{am9114,am9969,am9006,am501,am9222}. The results of this paper motivate connections between those and the lower order approach of this paper. Specifically it is of interest to identify the cases that are representable in terms of lower order semantics. The groupoidal approach of this paper is also extended to the higher order approaches in a separate paper.

A groupoid $S$ is \emph{tolerance trivial} if every definable compatible tolerance on it is a congruence. Key results can be found in \cite{sandor1991,chtol1991}. This concept extends to all algebras including the AR, AP, EUPGB and algebras of local approximations. In relation to knowledge interpretation, tolerance triviality amounts to a \emph{self organizing} aspect of knowledge. In other words, much less computational effort would be required to impose an interpretation on the semantics. This aspect is also explored in concrete terms by the first author in a forthcoming paper in the frameworks proposed in this research.

\bibliographystyle{model1b-num-names}
\bibliography{algroughf690}

\begin{thebibliography}{68}
\expandafter\ifx\csname natexlab\endcsname\relax\def\natexlab#1{#1}\fi
\providecommand{\bibinfo}[2]{#2}
\ifx\xfnm\relax \def\xfnm[#1]{\unskip,\space#1}\fi
%Type = Inproceedings
\bibitem[{Allen(2019)}]{alp2019}
\bibinfo{author}{P.~Allen}, \bibinfo{title}{{Show And Tell Software To Explore
  Maori Ways Of Communicating Mathematically}}, in:
  \bibinfo{editor}{J.~Subramanian}, et~al. (Eds.),
  \bibinfo{booktitle}{{Proceedings of MES10}}, \bibinfo{publisher}{Mathematics
  Education Society}, \bibinfo{year}{2019}, pp. \bibinfo{pages}{207--211}.
%Type = Book
\bibitem[{Burkhardt et~al.(2017)Burkhardt, Seibt, Imaguire and
  Gerogiorgakis}]{ham2017}
\bibinfo{editor}{H.~Burkhardt}, \bibinfo{editor}{J.~Seibt},
  \bibinfo{editor}{G.~Imaguire}, \bibinfo{editor}{S.~Gerogiorgakis} (Eds.),
  \bibinfo{title}{{Handbook of Mereology}}, \bibinfo{publisher}{Philosophia
  Verlag}, \bibinfo{address}{Germany}, \bibinfo{year}{2017}.
%Type = Book
\bibitem[{Burmeister(2002)}]{bu}
\bibinfo{author}{P.~Burmeister}, \bibinfo{title}{{A Model-Theoretic Oriented
  Approach to Partial Algebras}}, \bibinfo{publisher}{Akademie-Verlag},
  \bibinfo{year}{1986, 2002}.
%Type = Incollection
\bibitem[{Cattaneo(2018)}]{gc2018}
\bibinfo{author}{G.~Cattaneo}, \bibinfo{title}{{Algebraic Methods for Rough
  Approximation Spaces by Lattice Interior--closure Operations}}, in:
  \bibinfo{editor}{A.~Mani}, \bibinfo{editor}{I.~D{\"u}ntsch},
  \bibinfo{editor}{G.~Cattaneo} (Eds.), \bibinfo{booktitle}{{Algebraic Methods
  in General Rough Sets}}, {Trends in Mathematics},
  \bibinfo{publisher}{Springer International}, \bibinfo{year}{2018}, pp.
  \bibinfo{pages}{13--156}.
%Type = Incollection
\bibitem[{Cattaneo and Ciucci(2009)}]{cc5}
\bibinfo{author}{G.~Cattaneo}, \bibinfo{author}{D.~Ciucci},
  \bibinfo{title}{{Lattices With Interior and Closure Operators and Abstract
  Approximation Spaces}}, in: \bibinfo{editor}{J.F. Peters}, et~al. (Eds.),
  \bibinfo{booktitle}{{Transactions on Rough Sets X, LNCS 5656}},
  \bibinfo{publisher}{Springer}, \bibinfo{year}{2009}, pp.
  \bibinfo{pages}{67--116}.
%Type = Incollection
\bibitem[{Cattaneo and Ciucci(2018)}]{gcd2018}
\bibinfo{author}{G.~Cattaneo}, \bibinfo{author}{D.~Ciucci},
  \bibinfo{title}{{Algebraic Methods for Orthopairs and induced Rough
  Approximation Spaces}}, in: \bibinfo{editor}{A.~Mani},
  \bibinfo{editor}{I.~D{\"u}ntsch}, \bibinfo{editor}{G.~Cattaneo} (Eds.),
  \bibinfo{booktitle}{{Algebraic Methods in General Rough Sets}},
  \bibinfo{publisher}{Birkhauser Basel}, \bibinfo{year}{2018}, pp.
  \bibinfo{pages}{553--640}.
%Type = Book
\bibitem[{Chajda(1991)}]{chtol1991}
\bibinfo{author}{I.~Chajda}, \bibinfo{title}{{Algebraic Theory of Tolerance
  Relations}}, \bibinfo{publisher}{Olomouc University Press},
  \bibinfo{year}{1991}.
%Type = Article
\bibitem[{Chajda and Langer(2013{\natexlab{a}})}]{icl2013}
\bibinfo{author}{I.~Chajda}, \bibinfo{author}{H.~Langer},
  \bibinfo{title}{{Groupoids Assigned to Relational Systems}},
  \bibinfo{journal}{Math Bohemica} \bibinfo{volume}{138}
  (\bibinfo{year}{2013}{\natexlab{a}}) \bibinfo{pages}{15--23}.
%Type = Article
\bibitem[{Chajda and Langer(2013{\natexlab{b}})}]{ichlweak2013}
\bibinfo{author}{I.~Chajda}, \bibinfo{author}{H.~Langer}, \bibinfo{title}{{Weak
  Lattices}}, \bibinfo{journal}{Italian Journal of Pure and Applied
  Mathematics}  (\bibinfo{year}{2013}{\natexlab{b}}) \bibinfo{pages}{125--140}.
%Type = Article
\bibitem[{Chajda et~al.(2015)Chajda, Langer and Sevcik}]{ichlps2015}
\bibinfo{author}{I.~Chajda}, \bibinfo{author}{H.~Langer},
  \bibinfo{author}{P.~Sevcik}, \bibinfo{title}{{An Algebraic Approach to Binary
  Relations}}, \bibinfo{journal}{Asian European J. Math} \bibinfo{volume}{8}
  (\bibinfo{year}{2015}) \bibinfo{pages}{1--13}.
%Type = Incollection
\bibitem[{Chakraborty and Samanta(2008)}]{chp3}
\bibinfo{author}{M.K. Chakraborty}, \bibinfo{author}{P.~Samanta},
  \bibinfo{title}{{On Extension of Dependency and Consistency Degrees of Two
  Knowledges Represented by Covering}}, in: \bibinfo{editor}{J.F. Peters},
  \bibinfo{editor}{A.~Skowron} (Eds.), \bibinfo{booktitle}{{Transactions on
  Rough Sets IX, LNCS 5390}}, \bibinfo{publisher}{Springer Verlag},
  \bibinfo{year}{2008}, pp. \bibinfo{pages}{351--364}.
%Type = Incollection
\bibitem[{Chorney(2018)}]{chos2018}
\bibinfo{author}{S.~Chorney}, \bibinfo{title}{{Digital Technology in Teaching
  Mathematical Competency: A Paradigm Shift}}, in:
  \bibinfo{editor}{A.~Kajander}, et~al. (Eds.), \bibinfo{booktitle}{{Advances
  in Mathematics Education}}, \bibinfo{publisher}{Springer International},
  \bibinfo{year}{2018}, pp. \bibinfo{pages}{245--256}.
%Type = Incollection
\bibitem[{Ciucci(2017)}]{cd2017}
\bibinfo{author}{D.~Ciucci}, \bibinfo{title}{{Back To The Beginnings: Pawlak'S
  Definitions of The Terms Information System and Rough Set}}, in:
  \bibinfo{editor}{G.~Wang}, et~al. (Eds.), \bibinfo{booktitle}{{Thriving Rough
  Sets}}, {Studies in Computational Intelligence 708},
  \bibinfo{publisher}{Springer International}, \bibinfo{year}{2017}, pp.
  \bibinfo{pages}{225--236}.
%Type = Book
\bibitem[{Davey and Priestley(2002)}]{dp2002}
\bibinfo{author}{B.A. Davey}, \bibinfo{author}{H.A. Priestley},
  \bibinfo{title}{{Introduction to Lattices and Order}},
  \bibinfo{publisher}{Cambridge University Press}, \bibinfo{edition}{second}
  edition, \bibinfo{year}{2002}.
%Type = Article
\bibitem[{Duda and Chajda(1977)}]{jc1977}
\bibinfo{author}{J.~Duda}, \bibinfo{author}{I.~Chajda}, \bibinfo{title}{{Ideals
  of Binary Relational Systems}}, \bibinfo{journal}{Casopis pro pestovani
  matematiki} \bibinfo{volume}{102} (\bibinfo{year}{1977})
  \bibinfo{pages}{280--291}.
%Type = Inproceedings
\bibitem[{D{\"u}ntsch and Or{\l}owska(2011)}]{ie2011}
\bibinfo{author}{I.~D{\"u}ntsch}, \bibinfo{author}{E.~Or{\l}owska},
  \bibinfo{title}{{An Algebraic Approach To Preference Relations}}, in:
  \bibinfo{booktitle}{{Relational and Algebraic Methods in Computer Science -
  {RAMICS'2011} Proceedings of}}, pp. \bibinfo{pages}{141--147}.
%Type = Article
\bibitem[{Freese et~al.(2002)Freese, Jezek, Jipsen, Markovic, Maroti and
  Mckenzie}]{fjjm}
\bibinfo{author}{R.~Freese}, \bibinfo{author}{J.~Jezek},
  \bibinfo{author}{J.~Jipsen}, \bibinfo{author}{P.~Markovic},
  \bibinfo{author}{M.~Maroti}, \bibinfo{author}{R.~Mckenzie},
  \bibinfo{title}{{The Variety Generated by Order Algebras}},
  \bibinfo{journal}{Algebra Universalis} \bibinfo{volume}{47}
  (\bibinfo{year}{2002}) \bibinfo{pages}{103--138}.
%Type = Incollection
\bibitem[{Ganter and Meschke(2011)}]{bgc12}
\bibinfo{author}{B.~Ganter}, \bibinfo{author}{C.~Meschke}, \bibinfo{title}{{A
  Formal Concept Analysis Approach to Rough Data Tables}}, in:
  \bibinfo{editor}{J.F. Peters}, et~al. (Eds.),
  \bibinfo{booktitle}{{Transactions on Rough Sets XIV LNCS 6600}}, volume
  \bibinfo{volume}{LNCS 6600}, \bibinfo{year}{2011}, pp.
  \bibinfo{pages}{37--61}.
%Type = Book
\bibitem[{Ganter and Wille(1999)}]{gw99}
\bibinfo{author}{B.~Ganter}, \bibinfo{author}{R.~Wille},
  \bibinfo{title}{{Formal Concept Analysis: Mathematical Foundations}},
  \bibinfo{publisher}{Springer}, \bibinfo{address}{Berlin/Heidelberg},
  \bibinfo{year}{1999}.
%Type = Book
\bibitem[{Gr{\"a}tzer(1998)}]{gra1998}
\bibinfo{author}{G.~Gr{\"a}tzer}, \bibinfo{title}{{General Lattice Theory}},
  \bibinfo{publisher}{Birkhauser}, \bibinfo{year}{1998}.
%Type = Book
\bibitem[{Gratzer and Wehrung(2014)}]{gra2014}
\bibinfo{editor}{G.~Gratzer}, \bibinfo{editor}{F.~Wehrung} (Eds.),
  \bibinfo{title}{{Lattice Theory: Special Topics and Applications Volume 1}},
  \bibinfo{publisher}{Birkhauser Basel}, \bibinfo{year}{2014}.
%Type = Article
\bibitem[{Gruszczy{\'n}ski and Varzi(2015)}]{rgac15}
\bibinfo{author}{R.~Gruszczy{\'n}ski}, \bibinfo{author}{A.~Varzi},
  \bibinfo{title}{{Mereology Then and Now}}, \bibinfo{journal}{Logic and
  Logical Philosophy} \bibinfo{volume}{24} (\bibinfo{year}{2015})
  \bibinfo{pages}{409--427}.
%Type = Book
\bibitem[{Jacobs et~al.(2016)Jacobs, Renandya and Power}]{jrp2016}
\bibinfo{author}{G.M. Jacobs}, \bibinfo{author}{W.A. Renandya},
  \bibinfo{author}{A.~Power}, \bibinfo{title}{{Simple, Powerful Strategies for
  Student Centered Learning}}, {Springer Briefs in Education},
  \bibinfo{publisher}{Springer Nature}, \bibinfo{year}{2016}.
%Type = Incollection
\bibitem[{J{\"a}rvinen(2007)}]{jj}
\bibinfo{author}{J.~J{\"a}rvinen}, \bibinfo{title}{{Lattice Theory for Rough
  Sets}}, in: \bibinfo{editor}{J.F. Peters}, et~al. (Eds.),
  \bibinfo{booktitle}{{Transactions on Rough Sets VI}}, volume
  \bibinfo{volume}{LNCS 4374}, \bibinfo{publisher}{Springer Verlag},
  \bibinfo{year}{2007}, pp. \bibinfo{pages}{400--498}.
%Type = Article
\bibitem[{J{\"a}rvinen et~al.(2012)J{\"a}rvinen, Pagliani and Radeleczki}]{jpr}
\bibinfo{author}{J.~J{\"a}rvinen}, \bibinfo{author}{P.~Pagliani},
  \bibinfo{author}{S.~Radeleczki}, \bibinfo{title}{{Information Completeness in
  Nelson Algebras of Rough Sets Induced by Quasiorders}},
  \bibinfo{journal}{Studia Logica}  (\bibinfo{year}{2012})
  \bibinfo{pages}{1--20}.
%Type = Article
\bibitem[{J{\"a}rvinen and Radeleczki(2014)}]{jjs2014}
\bibinfo{author}{J.~J{\"a}rvinen}, \bibinfo{author}{S.~Radeleczki},
  \bibinfo{title}{{Rough Sets Determined by Tolerances}},
  \bibinfo{journal}{International Journal of Approximate Reasoning}
  \bibinfo{volume}{55} (\bibinfo{year}{2014}) \bibinfo{pages}{1419--1438}.
%Type = Article
\bibitem[{J{\"a}rvinen and Radeleczki(2017)}]{jjsr2017}
\bibinfo{author}{J.~J{\"a}rvinen}, \bibinfo{author}{S.~Radeleczki},
  \bibinfo{title}{{Representing Regular Pseudocomplemented Kleene Algebras by
  Tolerance-Based Rough Sets}}, \bibinfo{journal}{Journal of The Australian
  Mathematical Society}  (\bibinfo{year}{2017}) \bibinfo{pages}{1--22}.
%Type = Article
\bibitem[{Jarvinen et~al.(2009)Jarvinen, Radeleczki and Veres}]{jrv09}
\bibinfo{author}{J.~Jarvinen}, \bibinfo{author}{S.~Radeleczki},
  \bibinfo{author}{L.~Veres}, \bibinfo{title}{{Rough Sets Determined by
  Quasiorders}}, \bibinfo{journal}{Order} \bibinfo{volume}{26}
  (\bibinfo{year}{2009}) \bibinfo{pages}{337--355}.
%Type = Article
\bibitem[{Jezek and Kepka(1978)}]{jj1978}
\bibinfo{author}{J.~Jezek}, \bibinfo{author}{T.~Kepka},
  \bibinfo{title}{{Quasitrivial and Nearly Quasitrivial Distributive Groupoids
  and Semigroups}}, \bibinfo{journal}{Acta Univ. Carolinae Math et Phys}
  \bibinfo{volume}{19} (\bibinfo{year}{1978}) \bibinfo{pages}{22--44}.
%Type = Article
\bibitem[{Jezek and Mcenzie(2001)}]{jjm}
\bibinfo{author}{J.~Jezek}, \bibinfo{author}{R.~Mcenzie},
  \bibinfo{title}{{Variety of Equivalence Algebras}}, \bibinfo{journal}{Algebra
  Universalis} \bibinfo{volume}{45} (\bibinfo{year}{2001})
  \bibinfo{pages}{211--219}.
%Type = Article
\bibitem[{Jezek and Quakenbush(1990)}]{jjq90}
\bibinfo{author}{J.~Jezek}, \bibinfo{author}{R.~Quakenbush},
  \bibinfo{title}{{Directoids: Algebraic Models of Up-Directed Sets}},
  \bibinfo{journal}{Algebra Universalis} \bibinfo{volume}{27}
  (\bibinfo{year}{1990}) \bibinfo{pages}{49--69}.
%Type = Article
\bibitem[{Kepka(1981)}]{kt1981}
\bibinfo{author}{T.~Kepka}, \bibinfo{title}{{Quasitrivial Groupoids and
  Balanced Identities}}, \bibinfo{journal}{Acta Univ. Carolinae Math et Phys}
  \bibinfo{volume}{22} (\bibinfo{year}{1981}) \bibinfo{pages}{49--64}.
%Type = Article
\bibitem[{Koslicki(2007)}]{katk06}
\bibinfo{author}{K.~Koslicki}, \bibinfo{title}{{Towards a Neo-Aristotelian
  Mereology}}, \bibinfo{journal}{dialectica} \bibinfo{volume}{61}
  (\bibinfo{year}{2007}) \bibinfo{pages}{127--159}.
%Type = Book
\bibitem[{Ljapin(1996)}]{lj}
\bibinfo{author}{E.S. Ljapin}, \bibinfo{title}{{Partial Algebras and Their
  Applications}}, \bibinfo{publisher}{Academic, Kluwer}, \bibinfo{year}{1996}.
%Type = Book
\bibitem[{Malcev(1971)}]{mal}
\bibinfo{author}{A.I. Malcev}, \bibinfo{title}{{The Metamathematics of
  Algebraic Systems -- Collected Papers}}, \bibinfo{publisher}{North Holland},
  \bibinfo{year}{1971}.
%Type = Article
\bibitem[{Mani(2009)}]{am105}
\bibinfo{author}{A.~Mani}, \bibinfo{title}{{Algebraic Semantics of
  Similarity-Based Bitten Rough Set Theory}}, \bibinfo{journal}{Fundamenta
  Informaticae} \bibinfo{volume}{97} (\bibinfo{year}{2009})
  \bibinfo{pages}{177--197}.
%Type = Article
\bibitem[{Mani(2011)}]{am99}
\bibinfo{author}{A.~Mani}, \bibinfo{title}{{Choice Inclusive General Rough
  Semantics}}, \bibinfo{journal}{Information Sciences} \bibinfo{volume}{181}
  (\bibinfo{year}{2011}) \bibinfo{pages}{1097--1115}.
%Type = Article
\bibitem[{Mani(2012)}]{am240}
\bibinfo{author}{A.~Mani}, \bibinfo{title}{{Dialectics of Counting and The
  Mathematics of Vagueness}}, \bibinfo{journal}{Transactions on Rough Sets}
  \bibinfo{volume}{XV} (\bibinfo{year}{2012}) \bibinfo{pages}{122--180}.
%Type = Incollection
\bibitem[{Mani(2013)}]{am909}
\bibinfo{author}{A.~Mani}, \bibinfo{title}{{Towards Logics of Some Rough
  Perspectives of Knowledge}}, in: \bibinfo{editor}{Z.~Suraj},
  \bibinfo{editor}{A.~Skowron} (Eds.), \bibinfo{booktitle}{{Intelligent Systems
  Reference Library dedicated to the memory of Prof. Pawlak ISRL 43 }},
  \bibinfo{publisher}{Springer Verlag}, \bibinfo{year}{2013}, pp.
  \bibinfo{pages}{419--444}.
%Type = Article
\bibitem[{Mani(2014)}]{am3930}
\bibinfo{author}{A.~Mani}, \bibinfo{title}{{Ontology, Rough Y-Systems and
  Dependence}}, \bibinfo{journal}{Internat. J of Comp. Sci. and Appl.}
  \bibinfo{volume}{11} (\bibinfo{year}{2014}) \bibinfo{pages}{114--136}.
  \bibinfo{note}{Special Issue of IJCSA on Computational Intelligence}.
%Type = Article
\bibitem[{Mani(2016{\natexlab{a}})}]{am9501}
\bibinfo{author}{A.~Mani}, \bibinfo{title}{{Algebraic Semantics of
  Proto-Transitive Rough Sets}}, \bibinfo{journal}{Transactions on Rough Sets}
  \bibinfo{volume}{XX} (\bibinfo{year}{2016}{\natexlab{a}})
  \bibinfo{pages}{51--108}.
%Type = Book
\bibitem[{Mani(2016{\natexlab{b}})}]{amdsc2016}
\bibinfo{author}{A.~Mani}, \bibinfo{title}{{Granular Foundations of the
  Mathematics of Vagueness, Algebraic Semantics and Knowledge Interpretation}},
  \bibinfo{publisher}{University of Calcutta},
  \bibinfo{year}{2016}{\natexlab{b}}.
%Type = Article
\bibitem[{Mani(2016{\natexlab{c}})}]{am9411}
\bibinfo{author}{A.~Mani}, \bibinfo{title}{{Probabilities, Dependence and Rough
  Membership Functions}}, \bibinfo{journal}{International Journal of Computers
  and Applications} \bibinfo{volume}{39} (\bibinfo{year}{2016}{\natexlab{c}})
  \bibinfo{pages}{17--35}.
%Type = Inproceedings
\bibitem[{Mani(2016{\natexlab{d}})}]{am6900}
\bibinfo{author}{A.~Mani}, \bibinfo{title}{{Pure Rough Mereology and
  Counting}}, in: \bibinfo{booktitle}{{WIECON,2016}},
  \bibinfo{publisher}{IEEXPlore}, \bibinfo{year}{2016}{\natexlab{d}}, pp.
  \bibinfo{pages}{1--8}.
%Type = Incollection
\bibitem[{Mani(2017{\natexlab{a}})}]{am9006}
\bibinfo{author}{A.~Mani}, \bibinfo{title}{{Approximations From Anywhere and
  General Rough Sets}}, in: \bibinfo{editor}{L.~Polkowski}, et~al. (Eds.),
  \bibinfo{booktitle}{{Rough Sets-2, IJCRS,2017}}, {LNAI 10314},
  \bibinfo{publisher}{Springer International},
  \bibinfo{year}{2017}{\natexlab{a}}, pp. \bibinfo{pages}{3--22}.
%Type = Incollection
\bibitem[{Mani(2017{\natexlab{b}})}]{am9204}
\bibinfo{author}{A.~Mani}, \bibinfo{title}{{Generalized Ideals and Co-Granular
  Rough Sets}}, in: \bibinfo{editor}{L.~Polkowski}, et~al. (Eds.),
  \bibinfo{booktitle}{{Rough Sets, Part 2, IJCRS,2017 }}, {LNAI 10314},
  \bibinfo{publisher}{Springer International},
  \bibinfo{year}{2017}{\natexlab{b}}, pp. \bibinfo{pages}{23--42}.
%Type = Incollection
\bibitem[{Mani(2017{\natexlab{c}})}]{am9114}
\bibinfo{author}{A.~Mani}, \bibinfo{title}{{Knowledge and Consequence in AC
  Semantics for General Rough Sets}}, in: \bibinfo{editor}{G.~Wang}, et~al.
  (Eds.), \bibinfo{booktitle}{{Thriving Rough Sets}}, volume
  \bibinfo{volume}{708} of \textit{\bibinfo{series}{{Studies in Computational
  Intelligence Series}}}, \bibinfo{publisher}{Springer International},
  \bibinfo{year}{2017}{\natexlab{c}}, pp. \bibinfo{pages}{237--268}.
%Type = Incollection
\bibitem[{Mani(2018{\natexlab{a}})}]{am501}
\bibinfo{author}{A.~Mani}, \bibinfo{title}{{Algebraic Methods for Granular
  Rough Sets}}, in: \bibinfo{editor}{A.~Mani},
  \bibinfo{editor}{I.~D{\"u}ntsch}, \bibinfo{editor}{G.~Cattaneo} (Eds.),
  \bibinfo{booktitle}{{Algebraic Methods in General Rough Sets}}, {Trends in
  Mathematics}, \bibinfo{publisher}{Birkhauser Basel},
  \bibinfo{year}{2018}{\natexlab{a}}, pp. \bibinfo{pages}{157--336}.
%Type = Article
\bibitem[{Mani(2018{\natexlab{b}})}]{am9969}
\bibinfo{author}{A.~Mani}, \bibinfo{title}{{Dialectical Rough Sets, Parthood
  and Figures of Opposition-I}}, \bibinfo{journal}{Transactions on Rough Sets}
  \bibinfo{volume}{XXI} (\bibinfo{year}{2018}{\natexlab{b}})
  \bibinfo{pages}{96--141}.
%Type = Incollection
\bibitem[{Mani(2018{\natexlab{c}})}]{am5019}
\bibinfo{author}{A.~Mani}, \bibinfo{title}{{Representation, Duality and
  Beyond}}, in: \bibinfo{editor}{A.~Mani}, \bibinfo{editor}{I.~D{\"u}ntsch},
  \bibinfo{editor}{G.~Cattaneo} (Eds.), \bibinfo{booktitle}{{Algebraic Methods
  in General Rough Sets}}, {Trends in Mathematics},
  \bibinfo{publisher}{Birkhauser Basel}, \bibinfo{year}{2018}{\natexlab{c}},
  pp. \bibinfo{pages}{459--552}.
%Type = Article
\bibitem[{Mani(2019)}]{am9222}
\bibinfo{author}{A.~Mani}, \bibinfo{title}{{High Granular Operator Spaces and
  Less-Contaminated General Rough Mereologies}}, \bibinfo{journal}{Forthcoming}
   (\bibinfo{year}{2019}) \bibinfo{pages}{1--77}.
%Type = Article
\bibitem[{Mao et~al.(2019)Mao, Hu and Yao}]{hmy2019}
\bibinfo{author}{H.~Mao}, \bibinfo{author}{M.~Hu}, \bibinfo{author}{Y.Y. Yao},
  \bibinfo{title}{{Algebraic Approaches To Granular Computing}},
  \bibinfo{journal}{Granular Computing}  (\bibinfo{year}{2019})
  \bibinfo{pages}{1--13}.
%Type = Incollection
\bibitem[{Milani(2019)}]{milani2019}
\bibinfo{author}{R.~Milani}, \bibinfo{title}{{Opening An Exercise: Mathematics
  Prospective Teachers Entering In Landscapes Of Investigation}}, in:
  \bibinfo{editor}{J.~Subramanian}, et~al. (Eds.),
  \bibinfo{booktitle}{{Proceedings of MES10}}, \bibinfo{publisher}{Mathemation
  Education Society}, \bibinfo{year}{2019}, pp. \bibinfo{pages}{605--614}.
%Type = Incollection
\bibitem[{Nagy et~al.(2017)Nagy, Mihalydeak and Aszalos}]{dtl2017}
\bibinfo{author}{D.~Nagy}, \bibinfo{author}{T.~Mihalydeak},
  \bibinfo{author}{L.~Aszalos}, \bibinfo{title}{{Similarity Based Rough Sets}},
  in: \bibinfo{booktitle}{{Rough Sets, IJCRS'2017}}, {LNAI 10314},
  \bibinfo{publisher}{Springer International}, \bibinfo{year}{2017}, pp.
  \bibinfo{pages}{94--107}.
%Type = Incollection
\bibitem[{Pagliani(2018)}]{pp2018}
\bibinfo{author}{P.~Pagliani}, \bibinfo{title}{{Three Lessons on the
  Topological and Algebraic Hidden Core of Rough Set Theory}}, in:
  \bibinfo{editor}{A.~Mani}, \bibinfo{editor}{I.~D{\"u}ntsch},
  \bibinfo{editor}{G.~Cattaneo} (Eds.), \bibinfo{booktitle}{{Algebraic Methods
  in General Rough Sets}}, {Trends in Mathematics},
  \bibinfo{publisher}{Springer International}, \bibinfo{year}{2018}, pp.
  \bibinfo{pages}{337--415}.
%Type = Book
\bibitem[{Pagliani and Chakraborty(2008)}]{ppm2}
\bibinfo{author}{P.~Pagliani}, \bibinfo{author}{M.~Chakraborty},
  \bibinfo{title}{{A Geometry of Approximation: Rough Set Theory: Logic,
  Algebra and Topology of Conceptual Patterns}}, \bibinfo{publisher}{Springer},
  \bibinfo{address}{Berlin}, \bibinfo{year}{2008}.
%Type = Article
\bibitem[{Pawlak(1982)}]{zp6}
\bibinfo{author}{Z.~Pawlak}, \bibinfo{title}{{Rough sets}},
  \bibinfo{journal}{Internat. J. of Computing and Information Sciences}
  \bibinfo{volume}{18} (\bibinfo{year}{1982}) \bibinfo{pages}{341--356}.
%Type = Book
\bibitem[{Pawlak(1991)}]{zpb}
\bibinfo{author}{Z.~Pawlak}, \bibinfo{title}{{Rough Sets: Theoretical Aspects
  of Reasoning About Data}}, \bibinfo{publisher}{Kluwer Academic Publishers},
  \bibinfo{address}{Dodrecht}, \bibinfo{year}{1991}.
%Type = Book
\bibitem[{Polkowski(2011)}]{lp2011}
\bibinfo{author}{L.~Polkowski}, \bibinfo{title}{{Approximate Reasoning by
  Parts}}, \bibinfo{publisher}{Springer Verlag}, \bibinfo{year}{2011}.
%Type = Article
\bibitem[{Radeleczki(1991)}]{sandor1991}
\bibinfo{author}{S.~Radeleczki}, \bibinfo{title}{{Compatible Tolerances on
  Groupoids}}, \bibinfo{journal}{Czech. Math. J} \bibinfo{volume}{41}
  (\bibinfo{year}{1991}) \bibinfo{pages}{436--445}.
%Type = Incollection
\bibitem[{Seibt(2017)}]{seibtj2015}
\bibinfo{author}{J.~Seibt}, \bibinfo{title}{{Transitivity}}, in:
  \bibinfo{editor}{H.~Burkhardt}, \bibinfo{editor}{J.~Seibt},
  \bibinfo{editor}{G.~Imaguire}, \bibinfo{editor}{S.~Gerogiorgakis} (Eds.),
  \bibinfo{booktitle}{{Handbook of Mereology}}, \bibinfo{publisher}{Philosophia
  Verlag}, \bibinfo{address}{Germany}, \bibinfo{year}{2017}, pp.
  \bibinfo{pages}{570--579}.
%Type = Book
\bibitem[{Skovsmose(2011)}]{sko2011}
\bibinfo{author}{O.~Skovsmose}, \bibinfo{title}{{An Invitation to Critical
  Mathematics Education}}, \bibinfo{publisher}{Sense Publishers},
  \bibinfo{address}{Netherlands}, \bibinfo{year}{2011}.
%Type = Article
\bibitem[{Snasel(1997)}]{sva}
\bibinfo{author}{V.~Snasel}, \bibinfo{title}{{Lambda Lattices}},
  \bibinfo{journal}{Math. Bohemica} \bibinfo{volume}{122}
  (\bibinfo{year}{1997}) \bibinfo{pages}{267--272}.
%Type = Phdthesis
\bibitem[{Urbaniak(2008)}]{ur}
\bibinfo{author}{R.~Urbaniak}, \bibinfo{title}{{Lesniewski's Systems of Logic
  and Mereology; History and Re-Evaluation}}, Ph.D. thesis, Department of
  Philosophy, Univ of Calgary, \bibinfo{year}{2008}.
%Type = Inproceedings
\bibitem[{Usiskin(2012)}]{usz2012}
\bibinfo{author}{Z.~Usiskin}, \bibinfo{title}{{What Does it Mean to Understand
  Some Mathematics?}}, in: \bibinfo{editor}{others} (Ed.),
  \bibinfo{booktitle}{{12th International Congress on Mathematical Education
  (AMESA)}}, \bibinfo{publisher}{AMESA}, \bibinfo{year}{2012}.
%Type = Article
\bibitem[{Varzi(1996)}]{av}
\bibinfo{author}{A.~Varzi}, \bibinfo{title}{{Parts, Wholes and Part-Whole
  Relations: The Prospects of Mereotopology}}, \bibinfo{journal}{Data and
  Knowledge Engineering} \bibinfo{volume}{20} (\bibinfo{year}{1996})
  \bibinfo{pages}{259--286}.
%Type = Article
\bibitem[{Vieu(2007)}]{vie}
\bibinfo{author}{L.~Vieu}, \bibinfo{title}{{On The Transitivity of Functional
  Parthood}}, \bibinfo{journal}{Applied Ontology} \bibinfo{volume}{1}
  (\bibinfo{year}{2007}) \bibinfo{pages}{147--155}.
%Type = Article
\bibitem[{Yao et~al.(2012)Yao, Zhang and Miao}]{yzm2012}
\bibinfo{author}{Y.Y. Yao}, \bibinfo{author}{N.~Zhang},
  \bibinfo{author}{D.~Miao}, \bibinfo{title}{{Set-Theoretic Approaches To
  Granular Computing}}, \bibinfo{journal}{Fundamenta Informaticae}
  \bibinfo{volume}{115} (\bibinfo{year}{2012}) \bibinfo{pages}{247--264}.

\end{thebibliography}

\end{document}